\documentclass[12pt]{article}
\usepackage{amsmath}
\usepackage{graphicx,epsfig}
\usepackage{enumerate}
\usepackage{natbib}
\usepackage{color}
\usepackage{ulem}
\usepackage{verbatim}
\newcommand{\mat}[1]{\text{\mbox{\boldmath$#1$}}}
\newcommand{\bm}[1]{\text{\mbox{\boldmath$#1$}}}

\newtheorem{theorem}{Theorem}

\newtheorem{corollary}{Corollary}
\newtheorem{lemma}{Lemma}

\usepackage{setspace}
\usepackage{fancyhdr}
\begin{document}
 \title{Adaptive False Discovery Rate Control for Heterogeneous Data}
 \author{Joshua D. Habiger \\
  Department of Biostatistics \\ Kansas University Medical Center \\
 3109 Rainbow Blvd \\ Kansas City, KS \\ email:jhabiger@kumc.edu \\ }
\maketitle

\noindent {\it Abstract:}
Efforts to develop more efficient multiple hypothesis testing procedures for false discovery rate (FDR) control have focused on incorporating an estimate of the proportion of true null hypotheses (such procedures are called adaptive) or exploiting heterogeneity across tests via some optimal weighting scheme.  This paper combines these approaches using a weighted adaptive multiple decision function (WAMDF) framework.  Optimal weights for a flexible random effects model are derived and a WAMDF that controls the FDR for arbitrary weighting schemes when test statistics are independent under the null hypotheses is given.  Asymptotic and numerical assessment reveals that, under weak dependence, the proposed WAMDFs provide more efficient FDR control even if optimal weights are misspecified.  The robustness and flexibility of the proposed methodology facilitates the development of more efficient, yet practical, FDR procedures for heterogeneous data.  To illustrate, two different weighted adaptive FDR methods for heterogeneous sample sizes are developed and applied to data. \\

\noindent%

\section{Introduction}\label{sec 1}

High  throughput technology routinely generates data sets that call for hundreds or thousands of null hypotheses to be tested simultaneously.  For example, in \cite{AndHab12}, RNA sequencing technology was used to measure the prevalence of bacteria living near the roots of wheat plants across $i = 1, 2, ..., 5$ treatment groups for each of $m = 1, 2, ..., M = 778$ bacteria, thereby facilitating the simultaneous testing of $778$ null hypotheses.  See Table \ref{data} for a depiction of the data, or see Section 8 for more details. See also \cite{Efr08, DudVan08, Efr10} for other, sometimes called, high-dimensional (HD) data sets.

In general,  multiple null hypotheses are simultaneously tested with a multiple testing procedure which, ideally, rejects as many null hypotheses as possible subject to the constraint that some global type 1 error rate is controlled at a prespecified level $\alpha$.  The false discovery rate (FDR) is the most frequently considered error rate in the HD setting.  It is loosely defined as the expected value of the false discovery proportion (FDP), where the FDP is the proportion of erroneously rejected null hypotheses, also called false discoveries, among rejected null hypotheses, or discoveries.  See \cite{Sar07} for other related error rates.  In their seminal paper, \cite{BenHoc95} showed that a step-up procedure based on the \cite{Sim86} line, henceforth referred to as the BH procedure, has FDR  $= \alpha a_0\leq \alpha$ under a certain dependence structure, where $a_0$ is the proportion of true null hypotheses.  Since then, much research has focused on developing more efficient procedures for FDR control.

One  approach seeks to control the FDR at a level nearer $\alpha$, as opposed to $\alpha a_0$.  For example, adaptive procedures in \cite{BenHoc00,Sto04, BenKri06, Gav09, LiaNet12} utilize an estimate of $a_0$ and typically have FDR that is greater than $\alpha a_0$ yet still less than or equal to $\alpha$.  \cite{Fin09} proposed nonlinear procedures that ``exhaust the $\alpha$'' in that, loosely speaking, their FDR converges to $\alpha$ under some least favorable configuration as $M$ tends to infinity.

Another approach aims to exploit heterogeneity across hypothesis tests.  \cite{Gen06}; \cite{BlaRoq08}; \cite{RoqVan09}; \cite{Pen11} proposed a weighted BH-type procedure, where weights are allowed to depend on the power functions of the individual tests or prior probabilities for the states of the null hypotheses. \cite{Sto07} considered a ``single thresholding procedure'' which allowed for heterogeneous data generating distributions.  \cite{SunCai09} and \cite{Hu10} provided methods for clustered data, where test statistics are heterogeneous across clusters but homogeneous within clusters, while \cite{SunMcL12} considered heteroscedastic standard errors.  Data in Table \ref{data} are heterogeneous because sample sizes $n_1, n_2, ..., n_M$ vary from test to test, with $n_m$ being as small as $6$ and as large as $911$.
\begin{table}[t!]\caption{\label{data} Depiction of the data in \cite{AndHab12}.  Shoot  biomass $x_i$ in grams for groups $i$ = 1, 2, ..., 5 was  0.86, 1.34, 1.81, 2.37, and 3.00, respectively. Row totals are in the last column. }
\centering
\begin{tabular}[h!]{ccccccc|c} \hline \hline
Bacteria ($m$) && $Y_{1m}$ & $Y_{2m}$ & $Y_{3m}$ & $Y_{4m}$ & $Y_{5m}$ & Total ($n_m$)\\ \cline{1-1} \cline{3-8} \cline{3-8} \cline{1-1}
1 && 0 & 1 & 1 & 0 & 5 & 7\\
2 && 9 & 2 & 0 & 0 & 3 & 14\\
$\vdots$&&$\vdots$&$\vdots$&$\vdots$&$\vdots$&$\vdots$ & $\vdots$\\
778 && 16 & 10 & 29 & 18 & 13 & 81 \\
\hline
\end{tabular}
\end{table}

Whatever the nature of the heterogeneity may be, recent literature suggests that it should not be ignored. \cite{RoeWas09} showed that weighted multiple testing procedures generally perform favorably over their unweighted counterparts, especially when the employed weights efficiently exploit heterogeneity. Further, \cite{SunMcL12} showed that procedures which ignore heterogeneity can produce lists of discoveries that are of little scientific interest.

The objective of this paper is to provide a general approach for exploiting heterogeneity without sacrificing efficient FDR control.  The idea is to combine adaptive FDR methods for exhausting the $\alpha$ with weighted procedures for exploiting heterogeneity using a decision theoretic framework. Sections \ref{sec 2} - \ref{sec 5} provide the general framework.  Section \ref{sec 2} introduces multiple decision functions (MDFs) and a random effects model that can accommodate many types of heterogeneity including, but not limited to, those mentioned above.  Tools which facilitate easy implementation of MDFs, such as weighted $p$-values, are also developed.  Section \ref{sec 3} derives optimal weights for the random effects model and Section \ref{sec 4} introduces an asymptotically optimal weighted adaptive multiple decision function (WAMDF) for asymptotic FDP control.  Section \ref{sec 5} provides a WAMDF for exact (nonasymptotic) FDR control.

Assessment in Sections \ref{sec 6} and \ref{sec 7} reveals that, under a weak dependence structure, WAMDFs dominate other MDFs even when weights are misspecified.  Specifically, Section \ref{sec 6} shows that the asymptotic FDP of a WAMDF is larger than the FDP of its unadaptive counterpart, yet less than or equal to the nominal level $\alpha$.  Sufficient conditions for ``$\alpha$-exhaustion'' are provided and shown to be satisfied in a variety of settings.  For example, unweighted adaptive MDFs in \cite{Sto04} and certain asymptotically optimal WAMDFs are $\alpha$-exhaustive.  In fact, $\alpha$-exhaustion is achieved even in a worst-case-scenario setting, where employed weights are generated independently of optimal weights.  Simulation studies in Section \ref{sec 7} demonstrate that WAMDFs are more powerful than competing MDFs as long as the employed weights are positively correlated with optimal weights, and only slightly less powerful in the worse-case-scenario weighting scheme.

Section \ref{sec 8} provides two different routes for implementing WAMDFs in practice and compares them to one another.  They are applied to the data in Table 1 and shown analytically and with simulation to perform better than competing unweighted procedures. Concluding remarks are in Section \ref{sec 9} and technical details are in the Supplemental Article.

\section{Background} \label{sec 2}
\subsection{Data}
Let $\mat{Z} = (Z_m, m\in\mathcal{M})$ for $\mathcal{M} = \{1, 2, ..., M\}$ be a random vector of test statistics with joint distribution function $F$ and let $\mathcal{F}$ be a model for $F$.  The basic goal is to test null hypotheses $\mat{H} = (H_m, m\in\mathcal{M})$ of the form $H_m:F\in\mathcal{F}_m$, where $\mathcal{F}_m\subseteq\mathcal{F}$ is a submodel for $\mathcal{F}$.  For short, we often denote the state of $H_m$ by $\theta_m = 1 - I(F\in\mathcal{F}_m)$, where $I(\cdot)$ is the indicator function, so that $\theta_m = 0 (1)$ means that $H_m$ is true(false), and denote the state of $\mat{H}$ by $\bm{\theta} = (\theta_m, m\in\mathcal{M})$. Let  $\mathcal{M}_0 = \{m\in\mathcal{M}:\theta_m =0\}$ and $\mathcal{M}_1 = \mathcal{M}\setminus\mathcal{M}_0$ index the set of true and false null hypotheses, respectively, and denote the number of true and false null hypotheses by $M_0 = |\mathcal{M}_0|$ and $M_1 = |\mathcal{M}_1|$, respectively.

To make matters concrete, we often consider a random effects model for $\mat{Z}$.  For  related models see \cite{EfrTib01, GenWas02, Sto03, Gen06, SunCai07,SunCai09, RoqVan09}.  In Model 1, heterogeneity across the $Z_m$'s is attributable to prior probabilities $\mat{p} = (p_m, m\in\mathcal{M})$ for the states of the $H_m$'s and parameters $\bm{\gamma} = (\gamma_m, m\in\mathcal{M})$, which we refer to as effect sizes for ease of exposition, although each $\gamma_m$ could merely index a distribution for $Z_m$ when $H_m$ is false.  See, for example, Section \ref{sec 8}.\\
{\textbf{Model 1.} Let  $(Z_m, \theta_m, p_m, \gamma_m), m\in\mathcal{M}$, be independent and identically distributed random vectors with support in $\Re \times\{0,1\}\times [0,1]\times\Re^+$  and with conditional distribution functions
$F(z_m| \theta_m, p_m, \gamma_m) = (1-\theta_m)F_0(z_m) + \theta_m F_1(z_m|\gamma_m)$
and
$F(z_m|p_m,\gamma_m) = (1-p_m)F_0(z_m) + p_mF_1(z_m|\gamma_m).$
Assume $F(\gamma_m,p_m) = F(\gamma_m)F(p_m)$, $Var(\gamma_m)<\infty$ and that $p_m$ has mean $1-a_0 \in (0,1)$.

Observe  that $Z_m$ has distribution function $F_0(\cdot)$ given $H_m:\theta_m = 0$ and has distribution function $F_1(\cdot|\gamma_m)$ otherwise. Here, parameters $\bm{\theta}$, $\mat{p}$, and $\bm{\gamma}$ are assumed to be random variables to facilitate asymptotic analysis, as in \cite{Gen06, BlaRoq08, BlaRoq09,RoqVan09,RoqVil11}. Analysis under Model 1 focuses on conditional distribution functions $F(\mat{z}|\bm{\theta},\bm{p},\bm{\gamma}) = \prod_{m\in\mathcal{M}}F(z_m|\theta_m,p_m,\gamma_m)$ and $F(\mat{z}|\bm{p},\bm{\gamma}) = \prod_{m\in\mathcal{M}}F(z_m|p_m,\gamma_m)$, and an expectation taken over $\bm{Z}$ with respect to these distributions is denoted by $E[\cdot|\bm{\theta},\bm{p},\bm{\gamma}]$ and $E[\cdot|\bm{p},\bm{\gamma}]$, respectively. 

\subsection{Multiple decision functions}

A  multiple decision function (MDF) framework is used to formally define a multiple testing procedure. For similar frameworks see \cite{GenWas04, Sto04, SunCai07, Pen11}. Let $\delta_m(Z_m;t_m)$ denote a decision function taking values in $\{0,1\}$, where $\delta_m = 1 (0)$ means that $H_m$ is rejected(retained).  A decision function depends functionally on data $Z_m$ and (possibly random) ``size threshold'' $t_m\in[0,1]$.  To illustrate, suppose that large values of $Z_m$ are evidence against $H_m:\theta_m = 0$ under Model 1.  Then we may define
\begin{equation}\label{model1example} \delta_m(Z_m;t_m) = I(Z_m\geq F_0^{-1}(1-t_m)).\end{equation}
Observe  that $E[\delta_m(Z_m;t_m) |\theta_m = 0] = 1 - F_0(F_0^{-1}(1-t_m)) = t_m$ so that $t_m$ indeed represents the size of $\delta_m$, hence the terminology ``size threshold''.  An MDF is denoted $\bm{\delta}(\mat{Z};\mat{t}) = [\delta_m(Z_m;t_m), m\in\mathcal{M}]$, where $\mat{t} = (t_m, m\in\mathcal{M})$ is called a threshold vector. If $t_m = \alpha/M$ for each $m$ then $\bm{\delta}(\mat{Z};\mat{t})$ represents the well-known Bonferroni procedure.

Assume that, for each $m$, $t_m\mapsto\delta_m(Z_m;t_m)$ is nondecreasing and right continuous with $\delta_m = 0(1)$ whenever $t_m = 0 (1)$, almost surely, and that $t_m\mapsto E[\delta_m(Z_m;t_m)]$ is continuous and strictly increasing for $t_m\in(0,1)$, with $E[\delta_m(Z_m;t_m)]= t_m$ whenever $m\in\mathcal{M}_0$. These assumptions are referred to as the nondecreasing-in-size (NS) assumptions and are satisfied, for example, under Model 1 for decision functions defined as in (\ref{model1example}).  For additional details and examples see \cite{HabPen11, Pen11, Hab12}.

\subsection{Tools for implementation}
We break $\mat{t}$ down into the product of a positive valued weight vector $\bm{w} = (w_m, m\in\mathcal{M})$ satisfying $\bar{w} = M^{-1}\sum_{m\in\mathcal{M}} w_m = 1$ and an overall or average threshold $t$, $\mat{t} = t\mat{w}$.  First, weights are specified and then data $\mat{Z}=\mat{z}$ are collected, the overall threshold $t$ is computed, and the MDF $\bm{\delta}(\mat{z};t\mat{w})$ is computed.  If weights are based on Model 1, for example, then they are allowed to depend functionally on $\mat{p}$ and $\bm{\gamma}$.  The overall threshold is allowed to depend functionally on $\mat{z}$ and $\mat{w}$.

It is useful to exploit the link between weighted $p$-values and decision functions.  Define the (unweighted) $p$\textit{-value} statistic corresponding to $\delta_m$ by

$
P_m =\inf\{t_m\in[0,1]:\delta_m(Z_m;t_m) = 1\}.
$

This definition, see \cite{HabPen11,Pen11}, has the usual interpretation that $P_m$ is the smallest size $t_m$ allowing for $H_m$ to be rejected, and ensures that
$
\delta_m(Z_m;t_m) = I(P_m\leq t_m)
$
almost surely under the NS assumptions.  For example, it can be verified  that the $p$-value statistic corresponding to (\ref{model1example}) is $P_m = 1 - F_0(Z_m)$ and that $I(Z_m\leq F_0^{-1}(1-t_m)) = I(P_m\leq t_m)$ almost surely. See \cite{Hab12, HabPen14} for more details or for derivations of more complex $p$-values, such as the $p$-value for the local FDR statistic in \cite{EfrTib01, SunCai07} or for the optimal discovery procedure in \cite{Sto07}. Define the \textit{weighted $p$-value} statistic by

$
Q_m =\inf\{t:\delta_m(Z_m; tw_m)=1\}.
$

For $w_m$ fixed, and writing $t_m = t w_m$,
$$
P_m = \inf\{tw_m:\delta_m(Z_m;tw_m)=1\} = w_m\inf\{t:\delta_m(Z_m;tw_m) = 1\} = w_mQ_m
$$
almost surely.  Thus, a weighted $p$-value can be computed by $Q_m = P_m/w_m$.  Hence,  we have established the almost surely equivalent expressions for a decision function under the NS assumptions:
\begin{equation}\label{form}
\delta_m(Z_m;t_m) = \delta_m(Z_m;tw_m) = I(P_m\leq t w_m) = I(Q_m\leq t).
\end{equation}

\section{Optimal weights} \label{sec 3}
 Though results regarding exact FDR control in Section \ref{sec 5} or asymptotic FDP control in Section \ref{sec 6.1} apply more generally (see assumptions (A3) and  (A4) - (A6), respectively), optimal weights in this paper are developed for Model 1.  We first derive optimal weights assuming that $t$ is fixed/known.

\subsection{Optimal fixed-t weights} \label{sec 3.1}
We consider $\delta(\bm{Z};\bm{t})$ and the constraint that $\bar{w} = 1$ is replaced with the constraint that $\bar{t} = t$, where $\bar{t} = M^{-1}\sum_{m\in\mathcal{M}}t_m$. As weights are allowed to depend on $\mat{p}$ and $\bm{\gamma}$ under Model 1, the focus is on the conditional expectation of $\delta_m(Z_m;t_m)$ denoted by 
$
G_m(t_m) \equiv E[\delta_m(Z_m;t_m)|\bm{p},\bm{\gamma}] = (1-p_m)t_m + p_m\pi_{\gamma_m}(t_m),
$
where $\pi_{\gamma_m}(t_m) = E[\delta_m(Z_m;t_m)|\theta_m = 1,\gamma_m]$ is the power function for $\delta_m$. As in \cite{Gen06, RoqVan09, Pen11}, assume power functions (as a function of $t_m$) are concave.
\begin{enumerate}
\item[(A1)] For  each $m\in\mathcal{M}$, $t_m\mapsto\pi_{\gamma_m}(t_m)$ is concave and twice differentiable for $t_m\in(0,1)$, with $\lim_{t_m\uparrow1}\pi_{\gamma_m}'(t_m) = 0$ and $\lim_{t_m\downarrow 0}\pi_{\gamma_m}'(t_m) =\infty$ almost surely, where $\pi_{\gamma_m}'(t_m)$ is the derivative of $\pi_{\gamma_m}(t_m)$ with respect to $t_m$.
\end{enumerate}
This concavity  condition is satisfied, for example, under monotone likelihood ratio considerations (\cite{Pen11}) and under the generalized monotone likelihood ratio (GMLR) condition in \cite{Cao13}.

Given $\mat{p}$,  $\bm{\gamma}$, and $t$, the goal is to maximize the expected number of correctly rejected null hypotheses

$\pi(\mat{t},\mat{p},\bm{\gamma}) \equiv E\left[\sum_{m\in\mathcal{M}}\theta_m\delta_m(Z_m;t_m) \Big| \bm{\gamma},\mat{p}\right] = \sum_{m\in\mathcal{M}}p_m\pi_{\gamma_m}(t_m)$ subject to the constraint that $\bar{t} = t$.
\begin{theorem}\label{opt solution}
Suppose that (A1)  is satisfied, and fix $t\in(0,1)$.  Then under Model 1 the maximum of $\pi(\mat{t},\mat{p}, \bm{\gamma})$ with respect to $\mat{t}$ subject to constraint $\bar{t} = t$ exists, is unique, and satisfies
\begin{equation}\label{derivative}
\pi_{\gamma_m}'(t_m) = k/p_m
\end{equation}
for every $m\in\mathcal{M}$ and some $k>0$.
\end{theorem}
\cite{Spj72} and \cite{Sto07} also derived expressions for optimal fixed-$t$ thresholds, but did not allow for the states of the $H_m$'s to be random.  Specifically, \cite{Spj72} proposed maximizing $\sum_{m \in\mathcal{M}}\pi_{\gamma_m}(t_m)$ (see \cite{RoeWas09} for an illustration in the normal distribution setting) while \cite{Sto07} proposed maximizing $\sum_{m \in\mathcal{M}}\theta_m\pi_{\gamma_m}(t_m)$.

The important quantity in \eqref{derivative} is the constant $k$. In particular it suffices to find the unique value of $k$, say $k^*$, that satisfies $\bar{t} = t$.  For any value of $k$ denote the (unique) solution to (\ref{derivative}) in terms of $t_m$ as $t_m(k/p_m, \gamma_m)$, and take  $\mat{t}(k,\mat{p},\bm{\gamma}) = [t_m(k/p_m,\gamma_m), m\in\mathcal{M}]$. Then to compute weights
\begin{enumerate}
\item find the $k^*$ satisfying $\bar{t}_M(k^*,\mat{p},\bm{\gamma}) = t$,  where $\bar{t}_M(k,\mat{p},\bm{\gamma}) = M^{-1}\sum_{m\in\mathcal{M}}t_m(k/p_m,\gamma_m)$,
\item compute each optimal fixed-t weight
\begin{equation}\label{w star}
w_m(k^*,\mat{p},\bm{\gamma}) = \frac{t_m(k^*/p_m,\gamma_m)}{\bar{t}_M( k^*,\mat{p},\bm{\gamma})}.
\end{equation}
\end{enumerate}
We sometimes denote $w_m(k^*,\mat{p},\bm{\gamma})$ by $w_m^*$ and the vector of optimal fixed-t weights $\mat{w}(k^*,\mat{p},\bm{\gamma}) = [w_m(k^*,\mat{p},\bm{\gamma}),m\in\mathcal{M}]$ by $\mat{w}^* = (w_m^*,m\in\mathcal{M})$.

To better understand how the solution is found and related to the values of $p_m$, $\gamma_m$ and $t$ consider an example. \\
\textbf{Example 1.}
Suppose $Z_m|\gamma_m,\theta_m \sim N(\theta_m\gamma_m, 1)$ for $\gamma_m>0$ and consider testing $H_m:\theta_m = 0$.  Denote the standard normal cumulative distribution function and density function by $\Phi(\cdot)$ and $\phi(\cdot)$, respectively, and let $\bar{\Phi}(\cdot) = 1 - \Phi(\cdot)$.  Take $\delta_m(Z_m;t_m) = I(Z_m\geq \bar{\Phi}^{-1}(t_m))$.  The power function is
$
\pi_{\gamma_m}(t_m) = \bar{\Phi}(\bar{\Phi}^{-1}(t_m)-\gamma_m)
$
and has derivative
$\pi_{\gamma_m}'(t_m) = \frac{\phi(\bar{\Phi}^{-1}(t_m)-\gamma_m)}{\phi(\bar{\Phi}^{-1}(t_m))}.$
Setting the derivative equal to $k/p_m$ and solving yields
\begin{equation}\label{t example}
t_m(k/p_m, \gamma_m) = \bar{\Phi}\left(0.5\gamma_m + \log(k/p_m)/\gamma_m\right).
\end{equation}
The optimal fixed-$t$ threshold vector is computed as $\mat{t}(k^*,\mat{p},\bm{\gamma})$, where $k^*$  satisfies $\bar{t}_M(k^*,\mat{p},\bm{\gamma}) = t$, and the optimal fixed-$t$ weights are computed as in (\ref{w star}).

Observe in (\ref{t example}) that $t_i(k/p_i, \gamma_i) = t_{j}(k/p_{j},\gamma_{j})$ if  $\gamma_i=\gamma_{j}$ and $p_i=p_{j}$ regardless of $k$ and, consequently, the optimal fixed-$t$ weight vector is $\bm{1}$ for any $t$ when data are homogeneous. On the other hand, we see that $t_m(k/p_m,\gamma_m)$ is increasing in $p_m$ and hence $$w_m(k^*,\mat{p},\bm{\gamma}) = M\frac{t_m(k^*/p_m, \gamma_m)}{t_m(k^*/p_m,\gamma_m) + \sum_{j\neq m}t_j(k^*/p_j,\gamma_j)}$$ is increasing in $p_m$, as we might expect.  
\begin{figure}\centering
\epsfig{file = 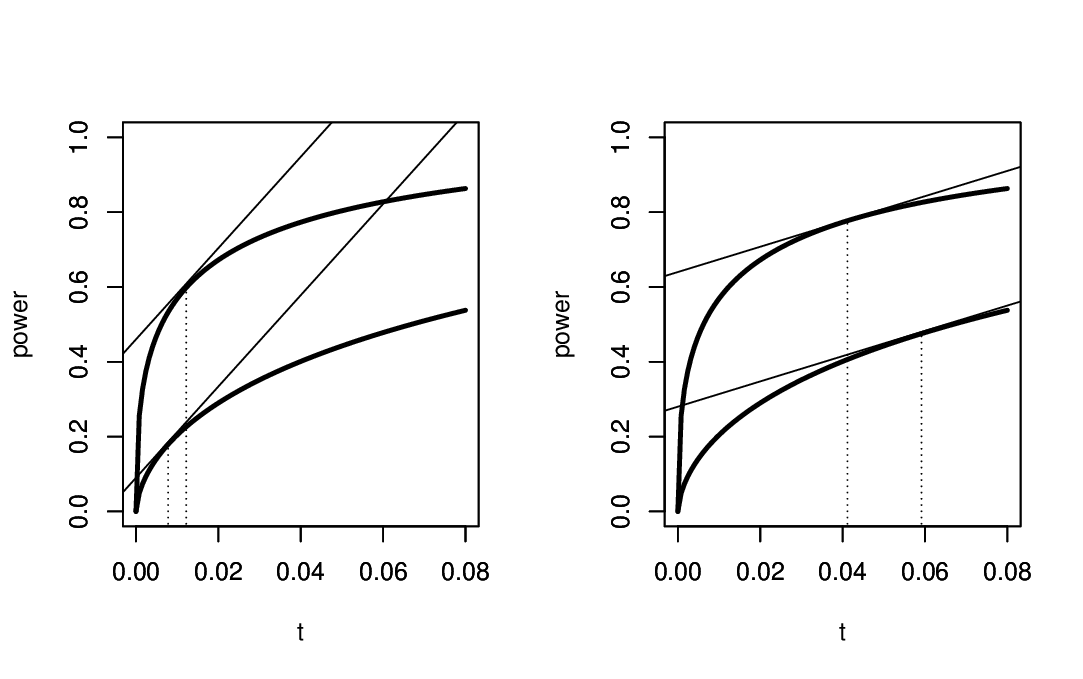, width = 5in}
\caption{\label{opt figure}A depiction of the optimal thresholds for $M = 2$ hypotheses tests when power functions vary
under constraint $0.5(t_1 + t_2) = 0.01$ (left) and $0.5(t_1 + t_2) = 0.05$ (right). }
\end{figure}

The relationship between $w_m(k^*,\mat{p},\bm{\gamma})$ and $\gamma_m$ is more complex. To illustrate, consider testing $M = 2$ null hypotheses and suppose $\gamma_1 = 1.5$, $\gamma_2 = 2.5$, and $p_1=p_2 = 0.5$.  In Figure \ref{opt figure}, observe that for $t = 0.01$, $\bar{t}_M(k^*,\mat{p},\bm{\gamma}) = 0.01$ when $k^* = 6.1$, which gives $t_1(k^*/p_1,\gamma_1) = 0.003$, $t_2(k^*/p_2, \gamma_2) = 0.017$, $w_1^* =0.003/0.01 = 0.3$ and $w_2^* = 0.017/0.01 = 1.7$. Because $p_1=p_2$, the slopes of the power functions evaluated at $0.003$ and $0.017$, respectively, are equal; see equation (\ref{derivative}).  Now consider the fixed threshold $t = 0.05$.  Here $k^* =  1.7$, which leads to weights $w_1^* = 0.059/0.05 = 1.18$ and $w_2^* = 0.041/0.05 = 0.82$.  Thus, when $t = 0.01$, the hypothesis with the larger effect size is given more weight, but when $t = 0.05$ it is given less weight.  For a more detailed discussion on this phenomenon see \cite{Pen11}.  The important point is that the optimal fixed-$t$ weights are only implementable if $t$ is fixed or specified before data collection.

\subsection{Asymptotically optimal weights}
The overall threshold $t$ in Section \ref{sec 4} depends on data $\mat{Z}$ because it depends on the FDP estimator, which depends functionally on $\mat{Z}$; see (\ref{M0 hat}) and (\ref{FDRhat}). The idea in this subsection is to approximate the FDP estimator using $\mat{p}$ and $\bm{\gamma}$.  This allows $t$ to be approximated before data collection so that the optimal fixed-$t$ weights can be utilized.

The FDP ``approximator'' plugs $G_m(t_m(k/p_m,\gamma_m)) = E[\delta_m(Z_m; t_m(k/p_m,\gamma_m))|\mat{p},\bm{\gamma}]$ in for each $\delta_m$ in (\ref{M0 hat}) and (\ref{FDRhat}). Formally, write $\bar{G}_M(\mat{t}(k,\mat{p},\bm{\gamma})) =M^{-1}\sum_{m\in\mathcal{M}}G_m(t_m(k/p_m,\gamma_m))$ and define the FDP approximator by
$$\widetilde{FDP}_M(\mat{t}(k,\mat{p},\bm{\gamma})) = \frac{1-\bar{G}_M(\mat{t}(k,\mat{p},\bm{\gamma}))}{1-\bar{t}_M(k,\mat{p},\bm{\gamma})}\frac{\bar{t}_M(k,\mat{p},\bm{\gamma})}{\bar G_M(\mat{t}(k,\mat{p},\bm{\gamma}))}.$$
Now, the asymptotically optimal weights are computed as follows.

\noindent{\textbf{Weight selection procedure}: \it \label{approx t wproc} For $0<\alpha\leq 1 - p_{(M)}$, where $p_{(M)} = \max\{\mat{p}\}$,
\begin{enumerate}
\item[a.] get $k_M^* = \inf\left\{k: \widetilde{FDP}_M(\mat{t}(k,\mat{p},\bm{\gamma})) = \alpha\right\},$ and
\item[b.] for each $m\in\mathcal{M}$, compute $w_m^* = w_m(k_M^*,\mat{p},\bm{\gamma})$ as in (\ref{w star}).
\end{enumerate}
}
In Theorem \ref{exist} we find that the restriction $0<\alpha\leq 1 - p_{(M)}$ ensures that a solution to $\widetilde{FDP}_M(\mat{t}(k,\mat{p},\bm{\gamma})) = \alpha$ exists.  In practice, this restriction amounts to choosing $\alpha$ and $\mat{p}$ so that $0<\alpha\leq 1-p_m$ for each $m$.  That is, the prior probability that the null hypothesis is true should be at least $\alpha$, which is reasonable in practice.
\begin{theorem}\label{exist}
Under (A1) and Model 1, $k^*_M$ exists for $0<\alpha\leq 1-p_{(M)}$.
\end{theorem}

Observe that $\bar{t}_M(k^*_M,\mat{p},\bm{\gamma}) = t$ for some $t\in(0,1)$ so that indeed these weights could be viewed as optimal fixed-t weights.  However, here weight computation is based on the constraint $\widetilde{FDP}_M(\mat{t}(k_M^*,\mat{p},\bm{\gamma})) = \alpha$.  These weights are henceforth referred to as asymptotically optimal for reasons that will be formalized later.

\section{The procedure}\label{sec 4}
Now we are now in position to formally define the proposed adaptive threshold which, when used in conjunction with asymptotically optimal weights in $\bm{\delta}(\mat{Z};t\mat{w})$, yields the asymptotically optimal WAMDF.

\subsection{Threshold selection}
For the moment, let $\mat{w}$ be any fixed vector of positive weights satisfying $\bar{w} = 1$.  For brevity, we sometimes suppress the $Z_m$ in each $\delta_m$ and write $\delta_m(tw_m)$ and denote $\bm{\delta}(\mat{Z};t\mat{w})$ by $\bm{\delta}(t\mat{w})$.  Further, denote the number of
discoveries at $t\mat{w}$ by $R(t\mat{w}) = \sum_{m\in\mathcal{M}}\delta_m(tw_m)$.

We make use of an ``adaptive'' estimator of the FDP that utilizes an estimator of $M_0$ defined by
\begin{equation} \label{M0 hat}
\hat{M}_0(\lambda \mat{w}) = \frac{M - R(\lambda\mat{w}) + 1}{1-\lambda}
\end{equation}
for some fixed tuning parameter $\lambda\in(0,1)$.  This  estimator is essentially the weighted version of the estimator in \cite{Sto02} defined by $\hat{M}_0(\lambda\bm{1}) = [M - R(\lambda\bm{1})]/[1-\lambda]$.  For earlier work on the estimation of $M_0$, see \cite{Sch82}.   As outlined in \cite{Sto04} in the unweighted setting, the idea is that for $m\in\mathcal{M}_1$, $E[\delta_m(\lambda)]\leq 1$, but the inequality is relatively sharp if all tests have reasonable power, which should be the case for large enough $\lambda$.  Hence $$E[M - R(\lambda\bm{1})]  = \sum_{m\in\mathcal{M}} E[1-\delta_m(\lambda)] \geq \sum_{m\in\mathcal{M}_0} E[1-\delta_m(\lambda)] = (1-\lambda)M_0$$ and $E[\hat{M}_0(\lambda\bm{1})] \geq M_0$. That is, $\hat{M}_0$ is positively biased but the bias is minor.  Similar intuition applies for $\hat{M}_0(\lambda\mat{w})$. As in \cite{Sto04}, we add 1 to the numerator in (\ref{M0 hat}) to ensure that $\hat{M}_0(\lambda\mat{w})>0$ for finite sample results.

The adaptive FDP estimator is defined by
\begin{equation}\label{FDRhat}
\widehat{FDP}^{\lambda}(t\mat{w}) =
\frac{\hat{M}_0(\lambda\mat{w}) t}{\max\{R(t\mat{w}), 1\}}.%
\end{equation}
The adaptive threshold,  which essentially chooses $t$ as large as possible subject to the constraint that the estimate of the FDP is less than or equal to $\alpha$,  is defined by
\begin{equation}
\hat t_\alpha^\lambda = \sup\{0 \leq t \leq u: \widehat{FDP}^\lambda(t\mat{w})\leq \alpha\}.\label{t stop}
\end{equation}
We assume that $u$, the  upper bound for $\hat{t}_\alpha^\lambda$, and the tuning parameter $\lambda$ satisfy
\begin{itemize}
\item[(A2)] $\lambda \leq u \leq 1/w_{(M)}$,
\end{itemize}
where $w_{(M)}\equiv \max\{\mat{w}\}$.  This  ensures that $\hat t_\alpha^\lambda w_m\leq 1$ and $\lambda w_m\leq 1$ for every $m$.  For $\mat{w} = \bm{1}$ and $u = \lambda$ (which implies $\hat{t}_\alpha^\lambda\leq \lambda$), we recover the unweighted adaptive MDF for
finite FDR control in \cite{Sto04}.

In practice $\hat{t}_{\alpha}^\lambda$ can be  difficult to compute.  Alternatively, we can apply the original BH procedure to the weighted $p$-values at level $\alpha M/\hat{M}_0(\lambda\mat{w})$. Due to (\ref{form}), we can also use weighted $p$-values to estimate $M_0$ via $\hat{M}_0(\lambda\mat{w}) = [M - \sum_{m\in\mathcal{M}}I(Q_m\leq \lambda)+1]/[1-\lambda].$
This threshold selection procedure can be implemented as follows.

\noindent{\textbf{Threshold selection procedure}: \it Fix $\lambda$ and $u$ satisfying (A2).  Then
\begin{enumerate}
\item[a.] compute $Q_m = P_m/w_m$ and ordered weighted $p$-values via $Q_{(1)}\leq Q_{(2)}\leq \ldots \leq Q_{(M)}$.
\item[b.] If $Q_{(m)}>\alpha m/\hat{M}_0(\lambda\mat{w})$ for each $m$, set $j = 0$, otherwise  take $$j = \max\left\{m\in\mathcal{M}:Q_{(m)}\leq
    \alpha/\hat{M}_0(\lambda\mat{w})\right\}.$$
\item[c.] Get $\hat t_{\alpha}^{\lambda *} = \min\{j\alpha/\hat{M}_0(\lambda\mat{w}), u\}$ and reject $H_m$ if $Q_m\leq \hat t_{\alpha}^{\lambda *}$.
\end{enumerate}
}

The WAMDF implemented above is equivalent to $\bm{\delta}(\mat{Z};\hat t_{\alpha}^\lambda\mat{w})$ in that
\begin{equation}\label{implement eqn}
\delta_m(Z_m;\hat t_\alpha^{\lambda}w_m) = I(Q_m\leq  \hat{t}_{\alpha}^\lambda) = I(Q_m\leq \hat{t}_{\alpha}^{\lambda*})
\end{equation}
almost surely for each $m$, so both procedures reject the same set of null hypotheses.   The first equality in (\ref{implement eqn}) follows from (\ref{form}) and the last equality in (\ref{implement eqn}) is a consequence of Lemma 2 in \cite{Sto04}.

\subsection{The asymptotically optimal WAMDF}

The asymptotically optimal WAMDF is formally defined as $\bm{\delta}(\mat{Z};\hat t_{\alpha}^\lambda \mat{w}^*)$ for  $0<\alpha\leq 1 - p_{(M)}$ and $\lambda = \bar{t}_M(k_M^*,\mat{p},\bm{\gamma})$, where $k_M^*$ and $\mat{w}^*$ are defined as in the Weight Selection Procedure.  This particular choice of $\lambda$ ensures that the employed weights are indeed ``asymptotically optimal'' (see Theorem \ref{asymp opt w}) and additionally that (A2) is satisfied if we take $u = 1/w_{(M)}$.  Other values of $\lambda$ could be considered, as in Section \ref{sec 8}.  To implement the the asymptotically optimal WAMDF, we compute $\mat{w}^*$ using the Weight Selection Procedure, then choose $\lambda = \bar{t}_M(k_M^*,\mat{p},\bm{\gamma})$ and $u$ satisfying (A2), collect data $\mat{Z}=\mat{z}$, and compute
 $\bm{\delta}(\mat{z};\hat t_\alpha^\lambda\mat{w}^*)$ using the Threshold Selection Procedure.

To illustrate, consider  testing $M = 10$ null hypotheses under the setting outlined in Example 1, with $p_m = 0.5$ for $m = 1, 2, ..., 10$, $\gamma_m =2$ for $m = 1, 2, \ldots, 5$, $\gamma_m = 3$ for $m = 6, 7, \ldots, 10$, and $\alpha = 0.05$.  The goal is to test $H_m:\theta_m = 0$ with decision functions $\delta_m(Z_m;t_m) = I(Z_m\geq \bar{\Phi}^{-1}(t_m))$ or their corresponding $p$-values $P_m = \bar{\Phi}(Z_m)$ and weighted $p$-values $Q_m = P_m/w_m$.  See Table \ref{illustrate} for summaries of parameters, weights, simulated data, $p$-values and weighted $p$-values.  The Weight Selection Procedure is broken down into 2 sub-steps and the Threshold Selection Procedure is split into three sub-steps.  To test these null hypotheses we
\begin{table}\center
\caption{\label{illustrate} A portion of the parameters, data, weights, $p$-values, and weighted $p$-values in columns 1 - 5, respectively.  Each row is sorted in ascending order according to $Q_1, Q_2, ..., Q_M$.}
\begin{tabular}{ccccccc} \hline \hline \scriptsize
$\theta_m$&$\gamma_m$ & $w_m^*$& $Z_m$ & $P_m$ & $Q_m$ & $0.05 m /\hat{M}_0$ \\ \hline
1 & 3 & 0.74 & 3.14 & 0.001 & 0.001 & 0.006 \\
1 & 2 & 1.26 & 2.55 & 0.005 & 0.005 & 0.012 \\
1 & 3 & 0.74 & 2.56 & 0.005 & 0.006 & 0.018 \\ \hline
1 & 2 & 1.26 & 1.47 & 0.070 & 0.062 & 0.024 \\
0 & 2 & 1.74 & 1.17 & 0.121 & 0.106 & 0.030 \\
$\vdots$ & $\vdots$ &$\vdots$ &$\vdots$ &$\vdots$ &$\vdots$ &  \\
0 & 3 & 0.74 & -0.60 & 0.724 & 0.844 & 0.061\\ \hline
\end{tabular}
\end{table}
\begin{enumerate}
\item[1a.] specify  $\bm{\gamma}$ (see column 2 of Table \ref{illustrate}), $\mat{p}$ and $\alpha$ and find $k_M^* = 2.52$.
\item[1b.] Compute  asymptotically optimal weights $w_m^* = w_m(k_M^*, \bm{p}, \bm{\gamma})$ as in (\ref{w star}).  See column 3 in Table
    \ref{illustrate}.
\item[2a.] Take  $\lambda = \bar{t}_M(k_M^*, \bm{p},\bm{\gamma}) = 0.028$ and $u = 1/1.26 = 0.79$.  Collect data $\mat{Z} = \mat{z}$ and compute and order weighted p-values (see columns 4 - 6 in Table \ref{illustrate}).
\item[2b.] Observe  that $Q_{(m)}\leq \alpha m/\hat{M}_0(\lambda\mat{w}^*)$ for $m = 3$ but not for $m = 4, 5, ..., 10$ and hence $\alpha
    j/\hat{M}_0(\lambda\mat{w}^*) = 0.05 \tfrac{3}{8.23} = 0.013$.
\item[2c.] Compute  $\hat{t}_{\alpha}^{\lambda*} = \min\{0.013, 0.79\} = 0.013$ and reject null hypotheses with weighted $p$-values 0.001, 0.005
    and 0.006 because they are less than 0.013.
\end{enumerate}

\section{Finite FDR control}\label{sec 5}
An upper bound  for the FDR is given for arbitrary weights satisfying $w_m>0$ for each $m$ and $\bar{w} = 1$.  The bound is computed under a dependence structure for $\mat{Z}$:
\begin{enumerate}
\item[(A3)] $(Z_m, m\in\mathcal{M}_0)$ are mutually independent and independent of $(Z_m, m\in\mathcal{M}_1)$.
\end{enumerate}
This structure has been utilized in \cite{BenHoc95, Gen06, Pen11, Sto04} to prove FDR control for unweighted unadaptive, weighted unadaptive, and unweighted adaptive procedures.  It is satisfied under Model 1 conditionally upon $(\bm{\theta},\mat{p},\bm{\gamma})$, but it is not limited to this setting.

To define the FDR,  let $V(t\mat{w}) = \sum_{m\in\mathcal{M}_0}\delta_m(tw_m)$ denote the number of erroneously rejected null hypotheses (false discoveries) at $t\mat{w}$, with $R(t\mat{w}) = \sum_{m\in\mathcal{M}}\delta_m(t\mat{w})$ the number of rejected null hypotheses.  Define the FDP at $t\mat{w}$ by
\begin{equation}\label{FDP}
FDP(t\mat{w}) = \frac{V(t\mat{w})}{\max\{R(t\mat{w}), 1\}}.
\end{equation}
The FDR at $t\mat{w}$ is defined by $FDR(t\mat{w}) = E[FDP(t\mat{w})]$, where the expectation is taken over $\mat{Z}$ with respect to an arbitrary $F\in\mathcal{F}$.

The bound is presented in Lemma \ref{bound}.  The focus is on the setting when $M_0\geq 1$ because the FDR is trivially 0 if $M_0 = 0$.   As in \cite{Sto04}, we force $\hat{t}_{\alpha}^\lambda\leq \lambda$ by taking $u = \lambda$ in (\ref{t stop}).  This facilitates the use of the Optional Stopping Theorem in the proof.
\begin{lemma}\label{bound}
Suppose $M_0\geq 1$ and that (A2) and (A3) are satisfied.   Then for $u = \lambda$,
\begin{equation}\label{lem1 eqn}
FDR(\hat t_\alpha^{\lambda}\mat{w}) \leq
\alpha\bar{w}_0\frac{1-\lambda}{1-\lambda\bar{w}_0}[1-(\lambda\bar{w}_0)^{M_0}]\leq\alpha\bar{w}_0\frac{1-\lambda}{1-\lambda\bar{w}_0},
\end{equation} where $\bar{w}_0 = M_0^{-1}\sum_{m\in\mathcal{M}_0}w_m$ is the mean of the weights from true null hypotheses.
\end{lemma}
\noindent Observe that $1-(\lambda\bar{w}_0)^{M_0} \leq 1 $ due to (A2).   Further, if $\mat{w} = \bm{1}$ then $\bar{w}_0 = 1$ and we recover Theorem 3 in \cite{Sto04} as a corollary.

If $\mat{w}\neq \bm{1}$, the bound in Lemma \ref{bound} is  not immediately applicable because $\mathcal{M}_0$, and consequently $\bar{w}_0$, is unobservable. One solution is to use an upper bound for $\bar{w}_0$ and adjust the ``$\alpha$'' at which the procedure is applied. This adjustment is described below.
\begin{theorem}\label{FDRthm}
If
$$\alpha^* = \alpha \frac{1}{w_{(M)}}\frac{1-\lambda w_{(M)}}{1-\lambda},$$
then under the conditions of Lemma \ref{bound}, $FDR(\hat{t}_{\alpha^*}^\lambda\mat{w})\leq \alpha$.
\end{theorem}
As $\bar{w}_0$ is typically less than or equal to 1, asymptotically, this $\alpha$ adjustment is not needed for large $M$.

\section{Asymptotic results} \label{sec 6}
We show that WAMDFs always reject more null hypotheses than their unadaptive counterparts, and provide sufficient conditions for asymptotic FDP control and $\alpha$-exhaustion.  These results are then used in the asymptotic analysis of the asymptotically optimal WAMDF.

To facilitate asymptotic analysis, denote weight vectors of length $M$ by $\mat{w}_M$ and the $m$th element of $\mat{w}_M$ by $w_{m,M}$.  Write the mean of the weights from true null hypotheses as $\bar{w}_{0,M}$.  Denote the adaptive FDP estimator in (\ref{FDRhat}) by $\widehat{FDP}^\lambda_M(t\bm{w}_M)$ and the FDP in (\ref{FDP}) by $FDP_M(t\bm{w}_M)$.  We also consider an unadaptive FDP estimator that uses $M$ in the place of an estimate of $M_0$, defined by
$$\widehat{FDP}_M^0(t\mat{w}_M) = \frac{M t}{\max\{R(t\mat{w}_M), 1\}}.$$
When necessary, we denote the tuning parameter in (\ref{M0 hat}) by $\lambda_M$ because,  as in the asymptotically optimal WAMDF where $\lambda_M = \bar{t}_M(k_M^*,\mat{p},\bm{\gamma})$, it may depend on $M$.

For asymptotic analysis, (A2) is redefined:
\begin{enumerate}
\item[(A2)]$\lambda_M\rightarrow \lambda\leq u = 1/k$ almost surely, where $k$ satisfies $\lim_{M\rightarrow \infty} w_{(M)}\leq k$ almost surely.
\end{enumerate}
The adaptive threshold in (\ref{t stop}) is denoted $\hat{t}_{\alpha,M}^\lambda$.   We find that (A2) is satisfied, for example,  under Model 1 and (A1) for the asymptotically optimal WAMDF.  The unadaptive threshold is defined by
$$\hat t_{\alpha,M}^0 = \sup\{0\leq t\leq u: \widehat{FDP}_M^0(t\mat{w}_M)\leq \alpha\}.$$

\subsection{Arbitrary weights}\label{sec 6.1}
Convergence criteria considered here are similar to criteria in \cite{Sto04, Gen06} and allow for weak dependence structures.  See \cite{Bil99}, \cite{Sto03}, or see Theorem \ref{optimal conditions} for examples.  For $u$ defined as in (A2) and $t\in(0,u]$, we assume the following.
\begin{enumerate}
\item[(A4)]$R(t\mat{w}_M)/{M} \rightarrow G(t)$ almost surely.
\item[(A5)]$V(t\mat{w}_M)/M \rightarrow a_0 \mu_0 t$ almost surely,  for $0<\mu_0<\infty$ and $0<a_0<1$, where $\bar{w}_{0,M} \rightarrow \mu_0$
    and  $M_0/M \rightarrow a_0$.
\item[(A6)] $t/G(t)$ is strictly increasing and continuous over (0,u) with $\lim_{t\downarrow 0}t/G(t) = 0$ and
    $\lim_{t\uparrow u}u/G(u) \leq 1.$
\end{enumerate}
Here $\mu_0$ is the asymptotic mean of the weights corresponding to true  null hypotheses and $a_0$ is the asymptotic proportion of true null hypotheses. The last condition is natural as it ensures that, asymptotically, the FDP is continuous and increasing in $t$ and takes on value 0, thereby ensuring that it can be controlled.   Writing $R(t\mat{w}_M)/M = \sum_{m\in\mathcal{M}}I(Q_m\leq t)/M$ via (\ref{form}), we see that (A4) corresponds to the assumption that the empirical process of the weighted $p$-values converges pointwise to $G(t)$ almost surely.

Asymptotic analysis for arbitrary weights focuses on comparing random  thresholds $\hat t_{\alpha,M}^\lambda$ and $\hat t_{\alpha,M}^0$ to their corresponding asymptotic (nonrandom) thresholds, which are based on the limits of the unadaptive and adaptive FDP estimators. Denote the pointwise limits of the  unadaptive FDP estimator, the adaptive FDP estimator, and the FDP by
$$FDP_\infty^0(t) = \frac{t}{G(t)}, \hspace{.1 in} FDP_\infty^\lambda(t) = \frac{1-G(\lambda)}{1-\lambda}\frac{t}{G(t)}, \mbox{\hspace{.1in} and \hspace{.1in} } FDP_{\infty}(t) = \frac{a_0\mu_0 t}{G(t)},$$ respectively (see Lemma S1 in the Supplemental Article for verification and details).   Define asymptotic unadaptive and asymptotic adaptive thresholds by, respectively,
$$t_{\alpha,\infty}^0 = \sup\{0\leq t\leq u:FDP_{\infty}^0(t)\leq \alpha\} \mbox{ and }t_{\alpha,\infty}^\lambda = \sup\{0\leq t\leq u:FDP_{\infty}^\lambda(t)\leq \alpha\}.$$

The unadaptive and adaptive thresholds converge to their asymptotic (nonrandom) counterparts, with the asymptotic adaptive threshold larger than the asymptotic unadaptive threshold.  As $E[\delta_m(tw_m)]$ is strictly increasing in $t$ for each $m$, it follows that the adaptive procedure leads to a higher proportion of rejected null hypotheses, asymptotically.  Our result generalizes Corollary 2 in \cite{Sto04}, which focused on the unweighted setting.
\begin{theorem}\label{asymptotic t}
Fix $\alpha\in(0,1)$.  Then under (A2) and (A4) - (A6), almost surely,
\begin{equation}\label{asymptotic t eqn}
\lim_{M\rightarrow\infty} \hat{t}_{\alpha,M}^0 = t_{\alpha,\infty}^0 \leq \lim_{M\rightarrow\infty}\hat{t}_{\alpha,M}^{\lambda} = t_{\alpha,\infty}^{\lambda}.
\end{equation}

\end{theorem}
\noindent

It is useful to formally describe the notion  of an $\alpha$-exhaustive MDF.  Loosely speaking, \cite{Fin09} referred to an unweighted multiple decision function, say $\bm{\delta}(\hat{t}_{\alpha,M}^*\bm{1}_M)$, as ``asymptotically optimal'' (we will use the terminology $\alpha$-exhaustive) if $FDR(\hat{t}_{\alpha,M}^*\bm{1}_M)\rightarrow \alpha$ under some least favorable distribution.  A Dirac Uniform (DU) distribution was shown to often be least favorable for the FDR in that, among all $F$s that satisfy $E[\delta_m(t)] = t$ for every $t\in[0,1]$ when $m\in\mathcal{M}_0$ and dependency structure (A3), $FDR(\hat{t}_{\alpha,M}^*\bm{1}_M)$ is the largest under a DU distribution.  In our notation, a DU distribution is any distribution satisfying $E[\delta_m(t)] = t$ if $m\in\mathcal{M}_0$ and $E[\delta_m(t)] = 1$ otherwise.  If (A4) - (A5) are satisfied, then $G(t) = a_0 \mu_0 t + (1-a_0)$ under a DU distribution for $t\leq u$.  Write this $G(t)$ as $G^{DU}(t)$.

To study the FDP of WAMDFs consider
\begin{equation}\lim_{M\rightarrow\infty}FDP_M(\hat{t}_{\alpha,M}^0\mat{w}_M)  \leq  \lim_{M\rightarrow\infty}FDP_M(\hat{t}_{\alpha,M}^\lambda\mat{w}_M)
\label{asymptotic FDP eqn1} \leq \alpha
\end{equation}
and three claims regarding these inequalities.
\begin{enumerate}
\item[(C1)] The first inequality in (\ref{asymptotic FDP eqn1}) is satisfied almost surely.
\item[(C2)] The second inequality in (\ref{asymptotic FDP eqn1}) is satisfied almost surely.
\item[(C3)] The second inequality in (\ref{asymptotic FDP eqn1}) is an equality almost surely under a DU distribution.
\end{enumerate}
Informally, Claim (C1) states that the FDP of  the WAMDF is asymptotically always larger than the FDP of its unadaptive counterpart and is referred to as the \textit{asymptotically less conservative} claim.  Claim (C2) states that the WAMDF has asymptotic FDP that is less than or equal to $\alpha$ and is referred to as the \textit{asymptotic FDP control} claim.  Claim (C3) is the $\alpha$\textit{-exhaustive} claim and states that the asymptotic FDP of the WAMDF is equal to $\alpha$ under a DU distribution.   Theorem \ref{asymptotic FDP} provides sufficient conditions for each claim.
\begin{theorem} \label{asymptotic FDP}
Fix $\alpha\in (0,1)$ and suppose that (A2)  and (A4) - (A6) are satisfied.  Then Claim (C1) holds.  Claim (C2) holds if, additionally, $\mu_0 \leq 1$.  Claim (C3) holds for $0<\alpha\leq FDP_\infty(u)$ if, additionally, $\mu_0 = 1$.
\end{theorem}

Asymptotic FDP control (C2)  and $\alpha$-exhaustion (C3) depend on the unobservable value of $\mu_0$, which necessarily depends on the weighting scheme at hand.
The next theorem is useful for verifying (C2) and/or (C3).
\begin{theorem}\label{EFDR}
Suppose that $(W_{m,M} , \theta_{m,M}), m \in\mathcal{M}$, are  identically distributed random vectors with support $\Re^+\times \{0,1\}$, and with $E[W_{m,M}]=1$ and $E[\theta_{m,M}] \in (0,1)$.  Take
$$\bar{W}_{0,M} = \frac{\sum_{m\in\mathcal{M}}(1-\theta_{m,M})W_{m,M}}{\sum_{m\in\mathcal{M}}(1-\theta_{m,M})}$$
whenever $\bm{\theta}_M \neq \bm{1}_M$ and $\bar{W}_{0,M} = 1$ otherwise. If $\bar{W}_{0,M}\rightarrow \mu_0$ almost surely, then $\mu_0\leq 1$ if $Cov(W_{m,M},\theta_{m,M})\geq 0$ and $\mu_0 = 1$ if $Cov(W_{m,M},\theta_{m,M}) = 0$.
\end{theorem}

\begin{corollary}\label{STS opt}
Suppose that (A4) - (A6) are satisfied and take $\mat{w}_M = \bm{1}_M$.  Then for any fixed $\lambda\in(0,1)$ and $0 < \alpha\leq a_0$,  Claims (C1) - (C3) hold.
\end{corollary}
This corollary suggests that the procedure in \cite{Sto04} is competitive with the $\alpha$-exhaustive nonlinear procedures in \cite{Fin09}.  That a DU distribution is the least favorable among such (unweighted) adaptive linear step-up procedures under our weak dependence structure is interesting; the search for least favorable distributions remains a challenging problem.  See \cite{Fin07, RoqVil11, Fin12}.

\subsection{Asymptotically optimal weights}
We verify that the conditions allowing for the WAMDF to provide less conservative asymptotic FDP control are satisfied under Model 1, even if the asymptotically optimal weights are perturbed or ``noisy''.
Weight vectors and elements of weight vectors are indexed by $M$ to facilitate asymptotic arguments, and, we sometimes write $\bar{t}_M(k_M^*) = \bar{t}_M(k_M^*,\mat{p},\bm{\gamma})$ for brevity.

Perturbed weights are simulated by multiplying each asymptotically optimal weight by a positive random variable $U_m$,
\begin{equation}
\tilde{w}_{m,M}(k_M^*,\mat{p},\bm{\gamma}) = U_m w_{m,M}(k_M^*,\mat{p},\bm{\gamma})
\label{perturbed w}
\end{equation}
for each $m$.  A perturbed weight is often denoted by $\tilde{w}_{m,M}$ and the vector of perturbed weights is denoted by $\tilde{\mat{w}}_M(k_M^*,\mat{p},\bm{\gamma})$ or $\tilde{\mat{w}}_M$.  To allow for (A2) to be satisfied, assume each triplet $(U_m, \gamma_m,p_m)$ has a joint distribution satisfying $0\leq U_mt_m(k_M^*/p_m,\gamma_m) \leq 1$ almost surely, and that $E[U_m|\mat{p},\bm{\gamma}] = 1$ for each $m$ so that perturbed weights have mean 1.  Here $\tilde{\mat{w}}_M = \mat{w}_M^*$ if $U_m = 1$ for each $m$ (almost surely).  Hence, results regarding perturbed weights immediately carry over to asymptotically optimal weights.
\begin{theorem} \label{optimal conditions}
Suppose that $\Pr(p_m\leq 1-\alpha) = 1$, take $\lambda_M = \bar{t}_M(k_M^*)$, and consider the perturbed weights $\tilde{\mat{w}}_M$.  Under  Model 1 and (A1), (A2) and (A4) - (A6) are satisfied and $\mu_0\leq 1$.  Hence the conditions of Theorem \ref{asymptotic t} are satisfied and (C1) and (C2) hold.
\end{theorem}

Next the notion of ``asymptotically optimal'' is formalized and some examples of $\alpha$-exhaustive weighting schemes are provided.  Asymptotically optimal weights are equivalent to optimal fixed-$t$ weights with $t = \bar{t}_M(k_M^*)$, while the asymptotically optimal WAMDF utilizes the asymptotic threshold $t_{\alpha,\infty}^\lambda$ (see Theorem \ref{asymptotic t}).
\begin{theorem} \label{asymp opt w}
Suppose that $\Pr(p_m\leq 1-\alpha) = 1$ and  take $\lambda_M = \bar{t}_M(k_M^*)$.  Then under Model 1 and (A1), $\bar{t}_M(k_M^*) \rightarrow t_{\alpha,\infty}^{\lambda}$ almost surely.
\end{theorem}

Two corollaries show that asymptotic $\alpha$-exhaustive FDP control is provided for a variety of weighing schemes.
\begin{corollary}\label{independent weights}
Under Model 1 and (A1) - (A2), if $\mat{w}_{M}$ are mutually  independent weights and independent of $\bm{\theta}_M$ with $E[w_{m,M}] = 1$, then (C1) - (C2) hold for $\alpha \in (0,1)$ and (C3) holds for $0 < \alpha\leq FDP_\infty(u)$.
\end{corollary}

The next setting arises in practice whenever the distributions of the $Z_m$'s from false nulls are heterogeneous, but heterogeneity attributable to prior probabilities for the states of the null hypotheses either does not exist or is not modeled.  For an illustration see Section \ref{sec 8}.  See also \cite{Spj72, Sto07, Pen11} for more on this type of heterogeneity.
\begin{corollary}\label{asymp opt p}
Suppose that the conditions of Theorem \ref{optimal conditions}  are satisfied and consider perturbed weights $\tilde{\bm{w}}_M$. If $p_i = p_j$ for every $i,j$, then (C3) holds for $0<\alpha\leq FDP_{\infty}(u)$.
\end{corollary}

The fact that $\alpha$-exhaustion need not be achieved when $p_i\neq p_j$ in Model 1 for the asymptotically optimal WAMDF, even though it is more powerful than competing MDFs, is noteworthy.  A similar phenomenon was observed in \cite{Gen06} in the unadaptive setting, and it was suggested that one potential route for improvement is to incorporate an estimate of $\mu_0$ into the procedure.  However, it is not clear how this objective could be accomplished without sacrificing FDP control, especially when weights may be perturbed.

\section{Simulation} \label{sec 7}

This section compares weighted adaptive MDFs to other MDFs in terms of power and FDP control via simulation.  In particular, for each of $K = 1000$ replications, we generate $Z_m\stackrel{i.i.d.}\sim N(\theta_m\gamma_m,1)$ for $m = 1, 2, ..., 1000$ and compute $\bm{\delta}(\hat{t}_{\alpha,M}^{\lambda}\mat{w}_M)$, $\bm{\delta}(\hat{t}_{\alpha,M}^{0}\mat{w}_M)$, $\bm{\delta}(\hat{t}_{\alpha,M}^{\lambda}\bm{1}_M)$, and $\bm{\delta}(\hat{t}_{\alpha,M}^{0}\bm{1}_M)$ as in Example 1, where $\alpha = 0.05$ and $\lambda_M = \bar{t}_M(k_M^*,\mat{p},\bm{\gamma})$. The average FDP and average correct discovery proportion (CDP) was computed over the $K$ replications for each procedure, where CDP $= \sum_{m\in\mathcal{M}_1}\delta_m/\max\{M_1, 1\}$.

In each simulation  experiment, $\gamma_m\stackrel{i.i.d.}\sim Un(1, a)$ for $a = 1, 3, 5$, $Un(1,a)$ the uniform distribution over $(1,a)$.  When $a = 1$ the effect sizes were identical, while when $a = 3$ or $a = 5$ they varied.  In Simulation 1, $p_m = 0.5$ for each $m$ and weighted procedures utilized asymptotically optimal weights.   In Simulation 2, weighted procedures used asymptotically optimal weights as before and the effect sizes varied as before, but $p_m\stackrel{i.i.d.}\sim Un(0,1)$.  Thus, though the procedure was optimally weighted and asymptotic FDP control was provided, the conditions of (C3) are no longer satisfied.  In Simulation 3, data were generated according to the same mechanism as in Simulation 2, but asymptotically optimal weights were perturbed via $U_m w_{m,M}(k_M^*,\mat{p},\bm{\gamma})$, where $U_m \stackrel{i.i.d.}\sim Un(0,2)$.  Simulation 4 represents a worst case scenario weighting scheme, in which weights were generated as $w_{m,M}\stackrel{i.i.d.}\sim Un(0,2)$.

Detailed results and discussions of simulations are in the supplemental materials.  The main point is that the WA procedure dominates all other procedures as long as the employed weights are at least positively correlated with the optimal weights, and it performs nearly as well as other procedures otherwise.  In particular, its FDP was less than or equal to $0.05$ in all simulations, as Theorem \ref{optimal conditions} stipulates.  Further, its average CDP was as large as or larger than the CDP of all other procedures in the first three simulations.  The WA procedure did have a slightly smaller average CDP than the UA procedure in the worst case scenario (Simulation 4), as one might expect.

\section{Implementation}\label{sec 8}
In practical applications parameters  $\mat{p}$ and $\bm{\gamma}$ in Model 1 are not (at least fully) observable and hence the asymptotically optimal WAMDF is not readily implementable.  However, these parameters can be estimated or specified based on reasonable assumptions if the nature of the heterogeneity is at least partially observable.  This section illustrates these two implementation approaches on the data in Table \ref{data} and discusses strengths and limitations of each.

\subsection{The Setup}
The goal is to test $H_m:\beta_{m}=0$ for each $m$, where $\beta_{m}$ is the regression coefficient for regressing $\mat{Y}_m = (Y_{1m}, Y_{2m}, ..., Y_{5m})^T$ on $\mat{x} = (x_1, x_2, ..., x_5)^T$ with the log-linear model log$(\mu_{im}) = \alpha_{m} + \beta_{m}x_i$ and where $Y_{im}$ are independent Poisson random variables with mean $\mu_{im}$. Let $N_m = \sum_{i =1}^5Y_{im}$ and $T_m = \sum_{i = 1}^5 x_iY_{im}$.  As per \cite{McCNel89}, we focus on the conditional distribution of $T_m|N_m = n_m$, which is free of the nuisance parameter $\alpha_{m}$.  Given $N_m = n_m$, $\mat{Y}_m$ has a multinomial distribution with mean $n_m \mat{p}(\beta_m)$ and covariance $n_m\left[diag(\mat{p}(\beta_m)) - \mat{p}(\beta_m)\mat{p}(\beta_m)^T\right]$, where $\mat{p}(a) = \left[\frac{\exp(x_1a)}{\sum_{i}\exp\{x_ia\}}, \frac{\exp(x_2a)}{\sum_{i}\exp\{x_ia\}}, ..., \frac{\exp(x_5a)}{\sum_{i}\exp\{x_ia\}}\right]^T$. Thus, the $Z$-score for $T_m = \mat{x}^T\mat{Y}_m$ is
$$Z_m = \left(\frac{T_m - n_m\mat{x}^T\mat{p}(0)}{\sqrt{n_m\mat{x}^T[diag(\mat{p}(0)) - \bm{p}(0)\bm{p}(0)^T]\mat{x}}}\right).$$

To facilitate Model 1 we consider the mixture model introduced in \cite{Hab15}, that assumes apriori that $\Pr(\beta_m = 0) = \pi_0$, $\Pr(\beta_m = \eta_1) = \pi_1$, and $\Pr(\beta_m = \eta_2) = \pi_2$ for some $\eta_1\neq \eta_2\neq 0$ and $\pi_0 + \pi_1 + \pi_2 = 1$.  Denote the mixing proportions by $\bm{\pi}$ and take $\bm{\eta} = (\eta_1, \eta_2)$.  Utilizing a normal approximation for the distribution of $Z_m$ results in normal mixture density for $Z_m|N_m = n_m$:
\begin{equation}\label{mixture}
f(z_m|n_m;\bm{\pi},\bm{\eta}) = \pi_0\phi\left(z_m;0,1\right) + \pi_1\phi\left(z_m;\mu(\eta_1,n_m),\sigma^2(\eta_1)\right) + \pi_2\phi\left(z_m;\mu(\eta_2,n_m),\sigma^2(\eta_2)\right),
\end{equation}
where $$
\mu(a, n_m) = \frac{\sqrt{n_m}\mat{x}^T[\mat{p}(a)-\mat{p}(0)]}{\sqrt{\mat{x}^T\left[diag(\mat{p}(0)) - \mat{p}(0)\mat{p}(0)^T\right]\mat{x}}} \mbox{ and }
\sigma^2(a) = \frac{\mat{x}^T\left[diag(\mat{p}(a)) - \mat{p}(a)\mat{p}(a)^T\right]\mat{x}}{\mat{x}^T\left[diag(\mat{p}(0))- \mat{p}(0)\mat{p}(0)^T\right]\mat{x}}.
$$
In the context of Model 1, $F_0 = \Phi$, $p_m = 1-a_0 = \pi_1 + \pi_2$, $\gamma_m = n_m$ and
$$F_1(z_m|\gamma_m) = F_1(z_m|n_m;\bm{\pi},\bm{\eta}) = \frac{\pi_1}{\pi_1+\pi_2}\Phi\left(\frac{z_m - \mu(\eta_1,n_m)}{\sigma(\eta_1)}\right) + \frac{\pi_2}{\pi_1+\pi_2}\Phi\left(\frac{z_m - \mu(\eta_2,n_m)}{\sigma(\eta_2)}\right).$$
Here $\gamma_m = n_m$ is not an unobservable effect size.  It is observable and indexes a mixture distribution for $Z_m$ when $H_m$ is false, which depends on the parameters $\bm{\pi}$ and $\bm{\eta}$.

The uniformly most powerful unbiased decision function is $\delta_m(Z_m;t_m) = I(|Z_m|\geq \Phi^{-1}(1-t_m/2))$, with power function
$$\pi_{n_m}(t_m) = F_1(\Phi^{-1}(t_m/2)|n_m;\bm{\pi},\bm{\eta}) + [1 - F_1(\Phi^{-1}(1-t_m/2)|n_m;\bm{\pi},\bm{\eta})]$$
To compute optimal fixed-t weights, first note that $\phi(\Phi^{-1}(t_m/2)) = \phi(\Phi^{-1}(1-t_m/2))$ so that the derivative of $\pi_{n_m}(t_m)$ with respect to $t_m$ is
\begin{eqnarray} \label{deriv}
\pi'_{n_m}(t_m) &\propto& \frac{f_1(\Phi^{-1}(t_m/2)|n_m;\bm{\pi},\bm{\eta}) + f_1(\Phi_0^{-1}(1-t_m/2)|n_m;\bm{\pi},\bm{\eta})}{\phi(\Phi^{-1}(t_m/2))}
\end{eqnarray}
Setting this derivative equal to $k/p_m = k/(\pi_1 + \pi_2)$ and solving for $t_m$ gives a collection of optimal fixed-$t$ thresholds.  Denote each such $t_m$ by $t_m(k,\bm{\pi},\bm{\eta},n_m)$.  Then, optimal fixed-t weights are computed as in
$w_m(k^*, \bm{\pi},\bm{\eta},n_m) = \frac{t_m(k^*,\bm{\pi},\bm{\eta},n_m)}{\bar{t}_M(k^*,\bm{\pi},\bm{\eta},\bm{n})}$ where $k^*$ satisfies $\bar{t}_M(k^*,\bm{\pi},\bm{\eta},\mat{n}) \equiv M^{-1}\sum_{m\in\mathcal{M}}t_m(k^*,\bm{\pi},\bm{\eta},n_m) = t$.

The five steps for implementing the WAMDF are:
\begin{enumerate}
\item[1a.] get $(\bm{\pi},\bm{\eta},n_m)$ for each $m$;
\item[1b.] compute $w_m^* = w_m(k_M^*,\bm{\pi},\bm{\eta},n_m)$ as in \eqref{asymp opt w};
\item[2a.] specify $\lambda$ and compute $Q_m = P_m/w_m^* = 2\bar\Phi(|z_m|)/w_m^*$;
\item[2b.] get $j = \max\{m:Q_{(m)}\leq \alpha m/\hat{M}_0(\lambda\bm{w}^*)\}$;
\item[2c.] get $\hat t_{\alpha}^{\lambda *} = \min\{j\alpha/\hat{M}_0(\lambda\mat{w}), \lambda\}$ and reject $H_m$ if $Q_m\leq \hat t_{\alpha}^{\lambda *}$.
\end{enumerate}
The parameters $\bm{\pi}$ and $\bm{\eta}$ are unobservable and hence must be estimated or specified.

\subsection{Parameter Estimation}
Parameters are estimated via maximum likelihood.  Specifically, assuming that $\mat{Y}_1, \mat{Y}_2, ..., \mat{Y}_M$ are independent conditionally upon $N_1, N_2,....,N_M$, then under \eqref{mixture}, the log likelihood is
$$l(\bm{\pi},\bm{\eta}) = \sum_{m = 1}^M log(f(z_m|n_m;\bm{\pi},\bm{\eta}))$$
and maximum likelihood estimates are found using the EM algorithm (\cite{Dem77}).  Results are summarized in Table \ref{MLEs}.  For more details on the EM algorithm and finite mixtures of normal distributions, see \cite{McLPee00} and see, for example, \cite{Ben09} for available software.
\begin{table}
\centering
\begin{tabular}{ccccc}
$\hat\pi_0$ & $\hat\pi_1$ & $\hat\pi_2$ & $\hat\eta_1$ & $\hat\eta_2$ \\ \hline
0.66 & 0.17 & 0.17 & -1.09 & 0.71
\end{tabular}
\caption{\label{MLEs} Maximum likelihood estimates for the model in \eqref{mixture}.}
\end{table}
For $\alpha = 0.05$ and $\lambda = 0.5$, the unweighted adaptive procedure resulted in 86 discoveries.  The weighted adaptive procedure with estimated weights as above (but modified via $\tilde w_m = [w_m^* + 0.1]/[M^{-1}\sum_{m}(w_m^* + 0.1)]$ to avoid impractically small weights) was applied for $\alpha = 0.05$ and $\lambda = 0.5$ and resulted in 85 discoveries. Of course, we cannot know the average power or FDR for the weighted and unweighed adaptive procedures based on this run of the experiment.

Some asymptotic results are readily available.  In particular, because $\hat{\bm{\pi}}$ and $\hat{\bm{\eta}}$ are maximum likelihood estimates, $\hat{\bm{\pi}}\rightarrow \bm{\pi}$ and $\hat{\bm{\eta}}\rightarrow \bm{\eta}$ as $M\rightarrow \infty$ almost surely.   Consequently, $w_m(k^*, \hat{\bm{\pi}},\hat{\bm{\eta}},n_m)\rightarrow w_m(k^*, \bm{\pi},\bm{\eta},n_m)$ as $M\rightarrow \infty$ almost surely.  See for example \cite{Ser80}, pg. 145 - 150.  Thus, this WAMDF is $\alpha$-exhaustive and asymptotically optimal under (\ref{mixture}).  A limitation of this approach is that it can be computationally intense, especially when $M$ is large.  Here parameters $\bm{\pi}$ and $\bm{\eta}$ must be estimated with an iterative procedure, a root finding algorithm is necessary to compute $t_m(k, \bm{\pi},\bm{\eta},n_m)$ for each $m$ and each value of $k$, and 3) a root-finding algorithm is necessary to find the $k^*$ corresponding to the asymptotically optimal weights.

\subsection{Parameter Specification}
One of the advantages of the WAMDF is that computationally simpler versions can be utilized with potentially little loss in efficiency and without sacrificing FDR control.  To illustrate, consider weights computed $w_m^* = t_m/\bar{t}$ where
\begin{equation}\label{simplest}
t_m =2\bar\Phi\left(0.5\bar{\Phi}^{-1}(\alpha/4)\left[\frac{\sqrt{n_m}}{\sqrt{n_\cdot}/M} +\frac{\sqrt{n_\cdot}/M}{\sqrt{n_m}}\right]\right)
\end{equation}
and where $\sqrt{n_\cdot} = \sum_m\sqrt{n_m}$. The WAMDF, with $\alpha = 0.05$, $\lambda = 0.5$, and $\tilde w_m = [w_m^* + 0.1]/[M^{-1} \sum_m (w_m^* + 0.1)]$ to safeguard against impractically small weights, was applied and resulted in 87 discoveries.

These weights utilized were justified as in (\ref{t example}), and by assuming that the average power and prior probability of $H_m$ being false is 1/2.  Specifically, $\mu(a,n_m)/\sigma(a) \propto \sqrt{n}_m$ and leads to approximate power functions as in Example 1 via $\pi_{\gamma_m}(t_m) = \bar\Phi(\bar\Phi^{-1}(t_m/2) - \gamma_m) = \bar\Phi(\bar\Phi^{-1}(t_m/2) - \gamma\sqrt{n_m})$ for $\gamma$ some tuning parameter.  Then, assume $p_m = 0.5$ and $\pi_{\bar\gamma}(t) = \bar\Phi(\bar\Phi^{-1}(t/2) - \bar{\gamma}) = 0.5$.  Approximating the FDR at $t$ when $p_m = 1/2$ and $\pi_{\bar\gamma}(t) =0.5$ with $FDR(t) = 0.5t/[0.5t + (1-0.5)\pi_{\bar\gamma}(t)]$, solving $FDR(t) = \alpha$ and $\pi_{\bar\gamma}(t) = 1/2$ simultaneously gives approximate fixed-$t$ threshold $t = \alpha/[2(1-\alpha)]\approx \alpha/2$ and $\bar{\gamma} = \bar\Phi^{-1}(t/2) \approx \bar{\Phi}^{-1}(\alpha/4)$. Taking the derivative of $\pi_{\gamma_m}(t_m) = \bar{\Phi}(\bar{\Phi}^{-1}(t_m/2) - \gamma_m)$ with respect to $t_m$ and setting it equal to $k/p$ and solving yields $log(k/p) = \bar\Phi^{-1}(t_m/2)\gamma_m - 0.5\gamma_m^2$, and
$$t_m = 2\bar\Phi\left(0.5\gamma_m + \frac{log(k/p)}{\gamma_m}\right). $$
Plugging $\bar{\Phi}^{-1}(t/2)\bar{\gamma} - 0.5\bar{\gamma}^2 = \bar{\gamma}^2 - 0.5\bar{\gamma}^2 = 0.5\bar{\gamma}^2$ in for $log(k/p)$, $\gamma_m = \sqrt{n_m} \gamma$, and $\bar{\gamma} = \gamma \sqrt{n_\cdot}/M$ here, we recover \eqref{simplest}.

These weights need not be asymptotically optimal. However, under \eqref{mixture} this WAMDF it is still $\alpha$-exhaustive (Corollary \ref{asymp opt p}) and simulation studies suggest that it is more efficient than its unweighted version even if these weights are only positively correlated with optimal weights.  The main advantage of this approach is that weights still exploit heterogeneity attributable to the $n_m$'s and are computationally simple.

The fact that weights are so simple allows for a simulation study to gauge the performance of the WAMDF.  In Simulation 5, for each of 1000 replications and $M = 1000$, we sampled $n_m$'s from the $n_m$'s in Table \ref{data} and generate $\theta_m \sim Bernoulli(p)$ and $Z_m\sim N(\gamma \sqrt{n_m} \theta_m, 1)$.  We considered all $p$-$\gamma$ combinations where $\gamma$ is chosen so that $\bar{\gamma} = \gamma M^{-1}\sqrt{n_\cdot} = 1.75, 2, 2.25$ and $p = 0.2, 0.5, 0.8$. For each replication and setting, the unweighted adaptive MDF was applied and the WAMDF was applied with $\alpha = 0.01, 0.05, 0.10$. The average FDP and CDP were recorded over the 1000 replications for each setting. Detailed results are in the supplemental materials.

Although the weights were based on some simplifying assumptions, the WAMDF was more powerful than in its unweighted counterpart even if $p = 0.2$ or $p = 0.8$, as long as the CDP was at least 0.2.  Further, the average FDP was always less than $\alpha$.  Our simplifying assumptions were made merely because they were the least informative and lead to the simplest weights. The resulting WAMDF outperformed its unweighted counterpart in most scenarios, and any other weighting schemes could be considered.  We leave more extensive methodological development of this nature as future work. The goal here was to demonstrate that the theory developed in the previous sections will be useful in developing WAMDFs that are simple and practical.

\section{Concluding remarks} \label{sec 9}

Efforts to improve upon  the original BH procedure have focused on controlling the FDR at a level nearer $\alpha$, or exploiting heterogeneity across tests.  We have combined these objectives using a weighted decision theoretic framework and showed that the resulting procedure is more powerful than procedures which only consider of them.  We have provided weighted adaptive multiple decision functions that satisfy the $\alpha$-exhaustive optimality criterion considered in \cite{Fin09}, but allow for further improvements via an optimal weighting scheme that incorporates heterogeneity.

The proposed WAMDFs are robust, and coupled with the flexibility of the WAMDF framework, allow for multiple testing procedures that exploit heterogeneity to be developed in a wide variety of settings, even when the nature and degree of heterogeneity is not fully observable or known.

The finite sample and  asymptotic results here are valid under independence and weak dependence conditions, respectively. \cite{BenYek01} showed that the unweighted unadaptive BH procedure provides (finite) FDR control under a certain positive dependence structure, and that it can be modified to control the FDR for arbitrary dependence.  One could study the performance of weighted adaptive procedures under other types of dependence, but obtaining finite sample analytical results for adaptive MDFs then appears to be very challenging.  See \cite{BlaRoq09, RoqVil11} for some results.  As for large sample results, \cite{Fan12} and \cite{DesSto12} provide techniques for transforming test statistics so that they are weakly dependent, and our WAMDF framework facilitates weak dependence.  Perhaps these transformed test statistics could be used in conjunction with our WAMDF, but this requires further development.

Other  estimators for $M_0$ could be considered.  For example, it is possible to use the unweighted estimator from \cite{Sto04} in the WAMDF, or to consider data dependent choices of the tuning parameter $\lambda$ as in \cite{LiaNet12}.  A more detailed assessment of $\hat{M}_0(\lambda\mat{w})$, though warranted, is beyond the scope of the present work.

\section*{Supplementary Materials}
Additional details and further discussion regarding simulations referred to in Sections \ref{sec 7} and \ref{sec 8}, and proofs of theorems, lemmas, and corollaries in Sections \ref{sec 3}, \ref{sec 5}, and \ref{sec 6} are in the supplemental materials.
\section*{Acknowledgements}
The author would like to thank referees and the AE for helpful suggestions.


\end{document}




\renewcommand{\baselinestretch}{2}

\markright{ \hbox{\footnotesize\rm Statistica Sinica: Supplement
}\hfill\\[-13pt]
\hbox{\footnotesize\rm
}\hfill }

\markboth{\hfill{\footnotesize\rm  Joshua D Habiger } \hfill}
{\hfill {\footnotesize\rm Adaptive FDR Control for Heterogeneous Data} \hfill}

\renewcommand{\thefootnote}{}
$\ $\par \fontsize{12}{14pt plus.8pt minus .6pt}\selectfont


 \centerline{\large\bf Adaptive False Discovery Rate Control for Heterogeneous Data}
\vspace{.25cm}
 \author{Joshua D Habiger}
\vspace{.4cm}
 \centerline{\it Department of Biostatistics \\ Kansas University Medical Center}
\vspace{.55cm}
 \centerline{\bf Supplementary Material}
\vspace{.55cm}
\fontsize{9}{11.5pt plus.8pt minus .6pt}\selectfont
\noindent
 This supplemental material contains proofs of theorems, lemmas and corollaries in \textit{Adaptive False Discovery Rate Control for Heterogeneous Data} and more details on simulation studies.
\par

\setcounter{section}{0}
\setcounter{equation}{0}
\def\theequation{S\arabic{section}.\arabic{equation}}
\def\thesection{S\arabic{section}}

\fontsize{12}{14pt plus.8pt minus .6pt}\selectfont

\section{Proofs of results in Section 3 \hspace{.1in}}

\noindent\textsc{Proof of Theorem 1.}
Setting up the Lagrangian
$$L(\mat{t},k) = \pi(\mat{t},\mat{p},\bm{\gamma}) - k\left[\left(\sum_{m\in\mathcal{M}}t_m\right) - Mt\right]$$
and taking derivative with respect to $t_m$ and setting it equal to 0 yields equation (5).   Now, recall we denote the solution to equation (5) with respect to $t_m$ by $t_m(k/p_m,\gamma_m)$ and observe $k\mapsto t_m(k/p_m,\gamma_m)$ is continuous and strictly decreasing in $k$ with $\lim_{k\rightarrow \infty}t_m(k/p_m,\gamma_m) = 0$ and $\lim_{k\downarrow 0}t_m(k/p_m,\gamma_m) = 1$ by (A1).  Thus, $\bar{t}_M(k,\mat{p},\bm{\gamma}) = M^{-1}\sum_{m\in\mathcal{M}}t_m(k/p_m,\gamma_m)$ is continuous and strictly decreasing in $k$ with  $\lim_{k\rightarrow \infty}\bar{t}_M(k,\mat{p},\bm{\gamma}) = 0$ and $\lim_{k\downarrow 0}\bar{t}_M(k,\mat{p},\bm{\gamma}) = 1$.  Hence, there exists a unique $k$ satisfying $\bar{t}_{M}(k,\mat{p},\bm{\gamma}) = t$ for any $t\in(0,1)$ and hence a unique collection $[t_m(k/p_m,\gamma_m),m\in\mathcal{M}]$.

To show that the solution is a maximum, it suffices to show that the sequence of the determinants of the principal minors of the bordered hessian matrix, evaluated at the solution, alternates in sign.  The $j$th principle minor of the bordered Hessian matrix is
$$\mat{H}_{j} = \left[\begin{array}{cc}0 & \mat{1}_j^T\\ \mat{1}_j& \mat{D}_j \end{array}\right]$$
where $\mat{D}_j$ is a $j\times j$ diagonal matrix with diagonal elements $d_m = \pi_{\gamma_m}''(t_m)$ and $\bm{1}_j$ is a vector of $1$s of length $j$.  Note that $d_m<0$ at the solution due to (A1).  Now, observe that $|\mat{H}_1| = -1 <0$ where $|\cdot|$ denotes the determinant, and for $j\geq 2$, we have the recursive relation
\begin{equation}
|\mat{H}_j| = d_j|\mat{H}_{j-1}| + (-1)^j\prod_{m=1}^{j-1}(-d_m).\label{H}
\end{equation}
Because $d_j<0$, for $j$ an even (odd) integer each expression on the righthand side of equation (\ref{H}) is positive (negative).  Hence $\{|\mat{H}_j|, j = 1, 2, ...\}$ alternates in sign. $\|$ \\

\lhead[\footnotesize\thepage\fancyplain{}\leftmark]{}\rhead[]{\fancyplain{}\rightmark\footnotesize\thepage}

\noindent\textsc{Proof of Theorem 2.} Observe that $\widetilde{FDP}_M(\mat{t}(k,\mat{p},\bm{\gamma}))$ is continuous in $k$ under (A1).  Hence, it suffices to show that $\widetilde{FDP}_M(\mat{t}(k,\mat{p},\bm{\gamma}))$ takes on values $0$ and $1-p_{(M)}$ by the Mean Value Theorem.  We first show that $$\lim_{k\downarrow 0}\widetilde{FDP}_M(\mat{t}(k,\mat{p},\bm{\gamma}))\geq 1 - p_{(M)}.$$  Observe that (A1) implies $t_m\leq \pi_{\gamma_m}(t_m)\leq 1$ for $t_m\in[0,1]$ and hence
\begin{equation}\label{ineq FDRM}
\bar{t}_M(k,\mat{p},\bm{\gamma}) \leq  \bar{G}_M(\mat{t}(k,\mat{p},\bm{\gamma})) \leq M^{-1}\left[\sum_{m\in\mathcal{M}}(1-p_m)t_m(k/p_m,\gamma_m) + p_m\right].
\end{equation}
The inequalities in (\ref{ineq FDRM}) imply
\begin{eqnarray*}
\widetilde{FDP}_M(\mat{t}(k,\mat{p},\bm{\gamma})) & = & \frac{\sum_{m\in\mathcal{M}}[1 - G_m(t_m(k/p_m,\gamma_m))]}{\sum_{m\in\mathcal{M}}[1-t_m(k/p_m,\gamma_m)]}\frac{\bar{t}_M(k,\mat{p}, \bm{\gamma})}{\bar G_M(\mat{t}(k,\mat{p},\bm{\gamma}))}\\
&\geq& \frac{\sum_{m\in\mathcal{M}}[1-p_m][1-t_m(k/p_m,\gamma_m)]}{\sum_{m\in\mathcal{M}}[1-t_m(k/p_m,\gamma_m)]}\frac{\bar{t}_M(k,\mat{p}, \bm{\gamma})}{\bar G_M(\mat{t}(k,\mat{p},\bm{\gamma}))}\\
& \geq & \left(1-p_{(M)}\right)\frac{\bar{t}_M(k,\mat{p},\bm{\gamma})}{\bar G_M(\mat{t}(k,\mat{p},\bm{\gamma}))},
\end{eqnarray*}
which converges to $1 - p_{(M)}$ as $k\downarrow 0$ if \begin{equation}
\label{ratio}
\frac{\bar{t}_M(k,\mat{p},\bm{\gamma})}{\bar G_M(\mat{t}(k,\mat{p},\bm{\gamma}))} \rightarrow 1
\end{equation}
as $k\downarrow0$.  To verify (\ref{ratio}), observe that $\lim_{k\downarrow 0}t_m(k/p_m,\gamma_m)= 1$ by (A1) and hence $\bar{t}_M(k, \mat{p},\bm{\gamma}) \rightarrow 1$ as $k\downarrow0$.  This, along with the inequalities in (\ref{ineq FDRM}), imply $\bar G_M(\mat{t}(k,\mat{p},\bm{\gamma}))\rightarrow 1$ as $k\downarrow 0$ and hence (\ref{ratio}) is satisfied.

Now if
\begin{equation} \label{c1}
\lim_{k\rightarrow \infty}\frac{\bar{t}_M(k,\mat{p},\bm{\gamma})}{\bar G_M(\mat{t}(k,\mat{p},\bm{\gamma}))}=0,
\end{equation} then by the first inequality in (\ref{ineq FDRM}) and the definition of $\widetilde{FDP}_M(\mat{t}(k,\mat{p},\bm{\gamma}))$
$$\widetilde{FDP}_M(\mat{t}(k,\mat{p},\bm{\gamma}))\leq \frac{\bar{t}_M(k,\mat{p},\bm{\gamma})}{\bar G_M(\mat{t}(k,\mat{p},\bm{\gamma}))}\rightarrow 0$$
as $k\rightarrow \infty$ and the proof would be complete. Hence, it suffices to show (\ref{c1}).
But because $t_m(k/p_m,\gamma_m)\downarrow 0$ as $k\rightarrow \infty$ and $\pi_{\gamma_m}'(t_m)\rightarrow \infty$ as $t_m\downarrow 0$ by (A1), we have $$\frac{\pi_{\gamma_m}(t_m(k/p_m,\gamma_m))}{t_m(k/p_m,\gamma_m)}\rightarrow \infty$$ as $k \rightarrow \infty$ by H{\^o}pital's rule.  Further for $a_m$, $b_m$,  $m\in\mathcal{M}$ any positive constants,
$$\frac{\sum_{m\in\mathcal{M}}a_m}{\sum_{m\in\mathcal{M}}b_m} = \sum_{m\in\mathcal{M}}\frac{a_m}{b_m}\left(\frac{b_m}{\sum_{m\in\mathcal{M}}b_m}\right)\geq \min\left\{\frac{a_m}{b_m}, m\in\mathcal{M}\right\}.$$
Hence,
$$A(k)\equiv \frac{\sum_{m\in\mathcal{M}}\pi_{\gamma_m}(t_m(k/p_m,\gamma_m))}{\sum_{m\in\mathcal{M}}t_m(k/p_m,\gamma_m)}\geq \min\left\{\frac{\pi_{\gamma_m}(t_m(k/p_m,\gamma_m))}{t_m(k/p_m,\gamma_m)}, m\in\mathcal{M}\right\}\rightarrow \infty$$
as $k\rightarrow \infty$ which implies
\begin{eqnarray*}
\frac{\bar{t}_M(k,\mat{p},\bm{\gamma})}{\bar{G}_M(\mat{t}(k,\mat{p},\bm{\gamma}))}&=& \left[\sum_{m\in\mathcal{M}}\frac{(1-p_m)t_m(k/p_m,\gamma_m)}{\bar{t}_M(k,\mat{p},\bm{\gamma})} + \frac{p_m\pi_{\gamma_m}(t_m(k/p_m,\gamma_m))}{\bar{t}_M(k,\mat{p},\bm{\gamma})}\right]^{-1}\\
&\leq&\left[M(1-p_{(M)}) + M p_{(1)}A(k)\right]^{-1}\rightarrow 0
\end{eqnarray*}
as $k\rightarrow \infty$, where $p_{(1)} \equiv \min\{\bm{p}\}$. $\|$

\section{Proofs of results in Section 5 \hspace{.1in}}
\setcounter{equation}{0}

\noindent\textsc{Proof of Lemma 1.}
The proof is based on the reverse martingale techniques in \cite{Sto04} for verifying FDR control in the unweighted adaptive setting and the proof of Theorem 9 in \cite{Pen11} for verifying FDR control in the weighted unadaptive setting.  The main additional challenge in the weighted setting is in verifying that $V(t\mat{w})$ is a martingale.  Also, the final proportion of the proof in \cite{Sto04} uses the fact that $V(t\bm{1})$ has a binomial distribution.  Here, $V(t\mat{w})$ is the sum of heterogeneous Bernoulli random variables and hence a Hoeffding inequality is necessary.

First, observe that because $u = \lambda$, $0\leq \hat{t}_\alpha^\lambda\leq \lambda$ by definition and that if $\hat{t}_\alpha^\lambda = 0$ then $FDR(\hat{t}_\alpha^\lambda\mat{w}) = 0$ trivially.  Let us focus on the setting where $0<\hat{t}_\alpha^\lambda\leq \lambda$.  By the definition of $\hat{t}_\alpha^\lambda$, $\widehat{FDP}^\lambda(\hat{t}_\alpha^\lambda\mat{w})\leq \alpha$ which gives $R(\hat{t}_\alpha^\lambda)\geq \hat{M}_0(\lambda\mat{w}) \hat{t}_\alpha^\lambda/\alpha$ by the definition of $\widehat{FDP}^\lambda(\cdot)$.  Hence,
\begin{eqnarray} 
FDR(\hat{t}_\alpha^\lambda\mat{w}) &=& E\left[\frac{V( \hat{t}_{\alpha}^\lambda\mat{w})}{R(\hat{t}_\alpha^\lambda\mat{w})}\right]\leq  E\left[\alpha\frac{1}{\hat{M}_0(\lambda\mat{w})}\frac{V(\hat{t}_\alpha^\lambda\mat{w})}{\hat{t}_\alpha^\lambda}\right]\\ &\leq& \label{FDReqn} E\left[\frac{\alpha}{\hat{M}_0(\lambda\mat{w})}\frac{V(\lambda\mat{w})}{\lambda}\right],
\end{eqnarray}
where (\ref{FDReqn}) is established as follows.  First, if $\hat{t}_\alpha^{\lambda} = \lambda$, it is true trivially. Now suppose that $0<\hat{t}_\alpha^{\lambda}<\lambda$. Define filtration $\mathcal{F}_t = \sigma\{\bm{\delta}(s\mat{w}), 0 < t\leq s \leq \lambda\}$ and observe that $\hat{t}_\alpha^{\lambda}$ is a stopping time with respect to $\mathcal{F}_{t}$ (with time running backwards).  Further, for $0< t\leq \lambda$,
$V(t\mat{w})/t$ is a reverse martingale with respect to $\mathcal{F}_{t}$.  This can be verified by noting that for $0 < s\leq t \leq \lambda$
\begin{eqnarray*}
E\left[\frac{V(s\mat{w})}{s}|\mathcal{F}_t\right] &=& \frac{1}{s}\sum_{m\in\mathcal{M}_0} E\left[\delta_m(sw_m)|\mathcal{F}_t\right] \\
&=&\frac{1}{s}\sum_{m\in\mathcal{M}_0} \delta_m(tw_m)E[\delta_m(sw_m)|\delta_m(tw_m)=1, \mathcal{F}_t] \\
&=&\frac{1}{s}\sum_{m\in\mathcal{M}_0} \delta_m(tw_m)E[\delta_m(sw_m)|\delta_m(tw_m)=1]\\
&=&\frac{1}{s}\sum_{m\in\mathcal{M}_0} \delta_m(tw_m)\frac{sw_m}{tw_m}\\
&=&\sum_{m\in\mathcal{M}_0}\frac{\delta_m(tw_m)}{t} \\
&=&\frac{V(t\mat{w})}{t},
\end{eqnarray*}
where first equality follows by the definition of $V(\cdot)$ and the second is due to the fact that $\delta_m(sw_m) = 0$ if $\delta_m(tw_m) = 0$ by the NS assumptions. The third equality is satisfied due to (A3).  The forth equality follows by the fact that $\Pr([\delta_m(sw_m)=1]\cap[\delta_m(tw_m)=1]) = E[\delta_m(sw_m)] = sw_m$ for $m\in\mathcal{M}_0$ and $s\leq \lambda$ under the NS assumptions and under (A2). The forth and fifth equalities follow from some algebra and the definition of $V(\cdot)$, respectively.  Hence, by the law of iterated expectation and the Optional Stopping Theorem \citep{Doo53}
\begin{eqnarray*}
E\left[\frac{\alpha}{\hat{M}_0(\lambda\mat{w})}\frac{V(\hat{t}_\alpha^\lambda\mat{w})}{\hat{t}_\alpha^\lambda}\right]& = &
E\left\{\frac{\alpha}{\hat{M}_0(\lambda\mat{w})}E\left[\frac{V(\hat{t}_\alpha^\lambda\mat{w})}{\hat{t}_\alpha^\lambda}|\mathcal{F}_\lambda\right]\right\}\\
&=&E\left[\frac{\alpha}{\hat{M}_0(\lambda\mat{w})}\frac{V(\lambda\mat{w})}{\lambda}\right].
\end{eqnarray*}
Hence, we have established (\ref{FDReqn}).

Now, note that $M - R(\lambda\mat{w}) = M_0 - V(\lambda\mat{w}) + [M_1 - \sum_{\mathcal{M}_1}\delta_m(\lambda w_m)]\geq M_0 - V(\lambda\mat{w})$. Further observe that $V\mapsto V(\lambda\mat{w})/[M_0-V(\lambda\mat{w}) + 1]$ is convex.  Hence, by Theorem 3 in \cite{Hoe56} and with $p = \lambda\bar{w}_0$
\begin{eqnarray*}
E\left[\frac{V(\lambda\mat{w})}{M_0 - V(\lambda\mat{w})+1}\right] &\leq& \sum_{k=0}^{M_0} \frac{k}{M_0 -k+1}\left(\begin{array}{c}M_0\\k\end{array}\right)p^k(1-p)^{M_0 - k}\\
&=&\frac{p}{1-p}(1-p^{M_0}).
\end{eqnarray*}
The last equality follows from basic calculations. Thus,
\begin{eqnarray*}
E\left[\alpha\frac{1}{\hat{M}_0(\lambda\mat{w})}\frac{V(\lambda\mat{w})}{\lambda}\right] & =
&E\left[\alpha\frac{(1-\lambda)}{M - R(\lambda\mat{w})+1}\frac{V(\lambda\mat{w})}{\lambda}\right] \\
&\leq & \alpha\frac{(1-\lambda)}{\lambda}E\left[\frac{V(\lambda\mat{w})}{M_0 - V(\lambda\mat{w})+1}\right]\\
&=&\alpha\frac{(1-\lambda)}{\lambda} \frac{p}{1-p}(1-p^{M_0}). \\
\end{eqnarray*}
The result follows by plugging $\lambda\bar{w}_0$ in for $p$ in the last expression. $\|$ \\

\noindent\textsc{Proof of Theorem 3.}
From Lemma 1 and because $\bar{w}_0\leq w_{(M)}$,
\begin{equation*}
FDR(\hat{t}_{\alpha^*}^\lambda\mat{w}) \leq\alpha^*\bar{w}_0\frac{1-\lambda}{1-\lambda\bar{w}_0} = \alpha
\frac{\bar{w}_0}{w_{(M)}}\frac{1-\lambda w_{(M)}}{1-\lambda\bar{w}_0} \leq \alpha. \|
\end{equation*}

\section{Proofs of results in Section 6}
\setcounter{equation}{0}

Before proving Theorem 4 the following Glivenko-Cantelli-type Lemma regarding the uniform convergence of the FDP estimators and the FDP is presented.  For similar results in the unweighted adaptive setting see Theorem 6 in \cite{Sto04} or see the proof of Theorem 2 in \cite{Gen06} for the weighted, but unadaptive, setting.  See also \cite{Fin09, Fan12} and references therein for additional results on almost sure convergence of the FDP.
\begin{lemma1} Fix $\delta \in (0,u)$.  Under (A2) and (A4) - (A6),
$$\sup_{\delta\leq t \leq u}|\widehat{FDP}_M^{0}(t\mat{w}_M) - FDP_{\infty}^0(t)|\rightarrow 0,$$
$$\sup_{\delta\leq t\leq u}|\widehat{FDP}_M^{\lambda}(t\mat{w}_M) - FDP_\infty^\lambda(t)|\rightarrow 0,$$ and
$$\sup_{\delta\leq t\leq u}|FDP_M(t\mat{w}_M) - FDP_\infty(t)|\rightarrow 0$$ almost surely.
\end{lemma1}
\begin{proof}
In what follows we denote $\max\{R(t\mat{w}_M), 1\}$ by $R(t\mat{w}_M)$ for short.  Observe $R(t\mat{w}_M)$ is nondecreasing in $t$ almost surely by the NS assumptions and $G(t)$ is strictly increasing in $t$ for $0\leq t\leq u$ by (A6).  Hence, for any $\delta\in (0,u),$
\begin{eqnarray*}\sup_{\delta\leq t\leq u}\left|\widehat{FDP}_M^{0}(t\mat{w}_M) - FDP_{\infty}^0(t)\right|&=&
\sup_{\delta\leq t\leq u}\left|\frac{t}{R(t\mat{w}_M)/M} -\frac{t}{G(t)}\right|\\ \\
 &&\hspace{-2.5in}=
\sup_{\delta\leq t\leq u}\left|\frac{t\left[G(t) - R(t\mat{w}_M)/M\right]}{G(t)R(t\mat{w}_M)/M}\right|
 \leq \frac{\sup_{\delta \leq t\leq u}\left|G(t) - R(t\mat{w}_M)/M\right|}{G(\delta)R(\delta\mat{w}_M)/M}\\
&&\hspace{-2.5in} \rightarrow
\frac{0}{G(\delta)^2} = 0 \\
\end{eqnarray*}
almost surely, where the numerator converges to 0 by the Glivenko-Cantelli Theorem and the denominator converges to $G(\delta)^2$ by (A4) and the Continuous Mapping Theorem.

As for the second claim, denote $\hat{a}_{0,M}^\lambda = \hat{M}_0(\lambda_M\mat{w}_M)/M$ and $a_{0,\infty}^\lambda = [1-G(\lambda)]/[1-\lambda]$.  Additionally observe
$$\widehat{FDP}_M^\lambda(t\mat{w}_M) = \hat{a}_{0,M}^\lambda \widehat{FDP}_M^0(t\mat{w})\mbox{ \hspace{.1in} and \hspace{.1in} } FDP^\lambda_\infty(t) = a_{0,\infty}^\lambda FDP_{\infty}^0(t),$$
Then using the triangle inequality
\begin{eqnarray*}
&&
\sup_{\delta\leq t\leq u}\left|\widehat{FDP}^\lambda_M(t\mat{w}_M) - FDP_\infty^\lambda(t)\right| 
=\sup_{\delta\leq t\leq u}\left|\hat{a}_{0,M}^\lambda \widehat{FDP}_M^0(t\mat{w}_M) - a_{0,\infty}^\lambda FDP_{\infty}^0(t)\right| \\
&&\hspace{.1in}\leq \left|\hat a_{0,M}^\lambda - a_{0,\infty}^\lambda\right|\times \sup_{\delta\leq t\leq u}\left|\widehat{FDP}_M^0(t\mat{w}_M)\right| 
+ \hspace{.1in} a_{0,\infty}^\lambda\times \sup_{\delta\leq t \leq u}\left|\widehat{FDP}_M^0(t\mat{w}_M) - \widehat{FDP}_{\infty}^0(t)\right| \\
&& \hspace{.1in} < 2\epsilon  + \epsilon,
 \end{eqnarray*}
where the last inequality is satisfied for all large enough $M$ for any $\epsilon>0$.  To verify the last inequality note that $\hat{a}_{0,M}^\lambda \rightarrow a_{0,\infty}^\lambda$ almost surely by (A2), (A4) and the Continuous Mapping Theorem, and hence $|\hat a_{0,M}^\lambda - a_{0,\infty}^\lambda|<\epsilon$ for all large enough $M$.  Further, for all large enough $M$, $$\sup_{\delta\leq t \leq u} \widehat{FDP}_M^0(t\mat{w}_M)< \sup_{\delta \leq t \leq u} FDP_\infty^0(t)\ + \epsilon \leq 2$$
by the first claim of the Lemma and (A6).  Additionally, $G(\lambda)\geq \lambda$ by (A6) and consequently $a_{0,\infty}^\lambda\leq 1$.  Lastly, $\sup_{\delta\leq t \leq u}|\widehat{FDP}_M^0(t\mat{w}_M) - FDP_{\infty}^0(t)|<\epsilon$ for all large enough $M$ by the first claim of the Lemma.

To prove the third claim, we first show that
\begin{eqnarray}
&&\sup_{\delta\leq t\leq u}\left|FDP_M(t\mat{w}_M) - FDP_\infty(t)\right| \nonumber \\
&&\leq\sup_{\delta\leq t\leq u}\left|\frac{V(t\mat{w}_M)}{R(t\mat{w}_M)} - \frac{a_0\mu_0 t}{R(t\mat{w}_M)/M}\right|
 + \sup_{\delta\leq t\leq u}\left|\frac{a_0\mu_0 t}{R(t\mat{w}_M)/M} - \frac{a_0\mu_0 t}{G(t)}\right| \nonumber\\
&&= \sup_{\delta\leq t\leq u}\frac{M}{R(t\mat{w})}\left|\frac{V(t\mat{w}_M)}{M} - a_0\mu_0 t\right|\label{S6}\\
 &&\hspace{.2in} + a_0\mu_0\sup_{\delta\leq t\leq u}\left|\widehat{FDP}_M^0(t\mat{w}_M) - FDP_\infty^0(t)\right| \label{S7}.
\end{eqnarray}
The inequality is a consequence of the triangle inequality and the definitions of $FDP_\infty(t)$ and $FDP_M(t\mat{w}_M)$.  The expression in (\ref{S6}) is verified by factoring out $R(t\mat{w}_M)/M$ in the first expression on the previous line while the expression in (\ref{S7}) follows from factoring out $a_0\mu_0$ in the second expression and by the definitions of $\widehat{FDP}^0_M(t\mat{w}_M)$ and $FDP_\infty^0(t)$.  Now, the quantity in (\ref{S7}) converges to 0 almost surely because $a_0\mu_0$ is bounded under (A5) and by the first claim of the Lemma.  To show that the first expression converges to 0 almost surely, first note for any $t\in(\delta,u]$, because $R(t\mat{w}_M)$ is nondecreasing in $t$, $R(t\mat{w}_M)/M>G(\delta/2)$ and hence that $$\frac{M}{R(t\mat{w}_M)}< \frac{1}{G(\delta/2)}$$ for all large enough $M$. Hence, if
\begin{equation}\label{S8}
\sup_{\delta\leq t \leq u}\left|\frac{V(t\mat{w}_M)}{M} - a_0\mu_0t\right|\rightarrow 0
\end{equation}
almost surely, then
$$\sup_{\delta\leq t\leq u}\frac{M}{R(t\mat{w})}\left|\frac{V(t\mat{w}_M)}{M} - a_0\mu_0 t\right|\leq \frac{\epsilon}{G(\delta/2)}$$
for all large enough $M$ and the proof would be completed since $\epsilon$ is arbitrary and $\delta$ is fixed.  To show (\ref{S8}), first observe that
$E[V(t\mat{w}_M)]/M_0 = \bar{w}_{0,M} t$ under the NS conditions. Also note that by the triangle inequality
\begin{eqnarray}\nonumber
\sup_{\delta\leq t \leq u}\left|\frac{V(t\mat{w}_M)}{M_0} - \mu_0 t\right|&\leq&
\sup_{\delta\leq t \leq u}\left|\frac{V(t\mat{w}_M)}{M_0} - \bar{w}_{0,M} t\right| + \sup_{\delta\leq t \leq u}\left|\bar{w}_{0,M}t - \mu_0 t\right|\\
&\leq& \sup_{\delta\leq t \leq u}\left|\frac{V(t\mat{w}_M)}{M_0} - \bar{w}_{0,M} t\right| + u\left|\bar{w}_{0,M} - \mu_0 \right| \rightarrow 0 \nonumber
\end{eqnarray}
almost surely, where the first quantity converges to 0 by the Glivenko-Cantelli Theorem and the second quantity converges to 0 because $\bar{w}_{0,M}\rightarrow \mu_0$ almost surely under (A5) and because $u\leq 1$.  Thus, again using the triangle inequality
\begin{eqnarray*}
&&\sup_{\delta\leq t \leq u}\left|\frac{V(t\mat{w}_M)}{M} - a_0\mu_0 t\right| = \sup_{\delta\leq t \leq u}\left|\frac{V(t\mat{w}_M)}{M_0}\left[\frac{M_0}{M} + a_0 - a_0\right] - a_0\mu_0 t\right|\\
&&\leq \left|\frac{M_0}{M} - a_0\right|\sup_{\delta\leq t\leq u}\left|\frac{V(t\mat{w}_M)}{M_0}\right| + a_0\sup_{\delta\leq t\leq u}\left|\frac{V(t\mat{w}_M)}{M_0} - \mu_0 t\right|\rightarrow 0 \end{eqnarray*}
almost surely, where the first quantity converges to 0 because $M_0/M\rightarrow a_0$ almost surely under (A5) and because $V(t\mat{w}_M)/M_0\leq 1$, while the second quantity converges to 0 because $a_0\leq 1$ and $V(t\mat{w}_M)/M_0\rightarrow \mu_0 t$.  Hence we have established (\ref{S8}).
\end{proof}

\noindent\textsc{Proof of Theorem 4.}  Let us first focus on the equalities.  Suppose that $t_{\alpha,\infty}^0<u$.  Then $FDP_{\infty}^0(t_{\alpha,\infty}^0) = \alpha$ by the definition of $t_{\alpha,\infty}^0$ and by (A6).  Additionally due to (A6), for any $\epsilon>0$ there exists a $0<\delta<\epsilon$ such that
$$ FDP_{\infty}^0(t_{\alpha,\infty}^0+\delta) < \alpha + \epsilon. $$
Now, Lemma S1 gives $\widehat{FDP}_M^0(t\mat{w}_M)< FDP_{\infty}^0(t_{\alpha,\infty}^0+\delta)$ for $0
\leq t<t_{\alpha,\infty}^0+\delta$ and all large enough $M$.  Hence,  this and (A6) imply $$\hat{t}_{\alpha,M}^0= \sup\left[0\leq t\leq u:\widehat{FDP}_M^0(t\mat{w}_M)\leq \alpha\right]\leq t_{\alpha,\infty}^0 +\delta<t_{\alpha,\infty}^0+\epsilon$$ for all large enough $M$.  Similar arguments give
$\hat{t}_{\alpha,M}^0>t_{\alpha,\infty}^0 -\epsilon$ for all large enough $M$.  Now if $t_{\alpha,\infty}^0 = u$ then
$$t_{\alpha,\infty}^0-\epsilon \leq \hat{t}_{\alpha,M}^0\leq t_{\alpha,\infty}^0 = u$$
for all large enough $M$.  Hence, $|\hat{t}_{\alpha,M}^0 - t_{\alpha,\infty}^0|<\epsilon$ for all large enough $M$ and we conclude $\hat{t}_{\alpha,M}^0\rightarrow t_{\alpha,\infty}^0$ almost surely.  As for the second equality, $FDP^\lambda_\infty(t) = a_{0,\infty}^\lambda FDP_\infty^0(t)$ is also continuous and strictly increasing by (A6) and consequently identical argument apply.  Thus $\hat{t}_{\alpha,M}^{\lambda} \rightarrow t_{\alpha,\infty}^{\lambda}$ almost surely. 

As for the inequality, note that (A6) implies $\lambda\leq G(\lambda)$ which implies
\begin{equation}\label{a0 ineq}
a_{0,\infty}^\lambda = \frac{1-G(\lambda)}{1-\lambda} \leq  1. \end{equation}
Hence,
\begin{equation}\label{asym FDP ineq}
FDP_{\infty}^\lambda(t) = a_{0,\infty}^\lambda FDP_\infty^0(t)\leq FDP_\infty^0(t)
\end{equation} for every $t\in(0,u]$.  This, (A6) and the definitions of $FDP_\infty^0(\cdot)$, $t_{\alpha,\infty}^0 $ and $t_{\alpha,\infty}^\lambda$ imply $t_{\alpha,\infty}^0 \leq t_{\alpha,\infty}^\lambda$. $\|$\\

\noindent\textsc{Proof of Theorem 5.}
By Lemma S1 and (A6), for $0<s<t\leq u$
\begin{eqnarray*}
FDP_M(t\mat{w}_M) - FDP_M(s\mat{w}_M)&>& \\
&& \hspace{-2in}a_0\mu_0 t/G(t) - a_0\mu_0 s/G(s) - 2\sup_{0\leq t\leq u}\left|FDP_M(t\mat{w}_M) - a_0\mu_0t/G(t)\right|\\
&&\hspace{-1in}\rightarrow  a_0\mu_0[t/G(t) - s/G(s)]> 0
\end{eqnarray*}
almost surely.  Claim (C1) is then a consequence of Theorem 4 and the Continuous Mapping Theorem.  To verify Claims (C2) and (C3), first observe that by the triangle inequality
\begin{eqnarray*}
&&|FDP_M(\hat{t}_{\alpha,M}^\lambda \mat{w}_M) - FDP_\infty(t_{\alpha,\infty}^\lambda)| \\
&\leq&|FDP_M(\hat{t}_{\alpha,M}^\lambda \mat{w}_M) - FDP_{\infty}(\hat t_{\alpha,M}^\lambda)|
+ |FDP_\infty(\hat{t}_{\alpha,M}^\lambda) - FDP_{\infty}(t_{\alpha,\infty}^\lambda)|.
\end{eqnarray*}
The first quantity converges to 0 almost surely by Lemma S1 and the second quantity converges to 0 almost surely by Theorem 4 and the Continuous Mapping Theorem.
Hence, $FDP_M(\hat{t}_{\alpha,M}^\lambda\mat{w}_M)\rightarrow FDP_\infty(t_{\alpha,\infty}^\lambda)$ almost surely.  Thus to prove Claims (C2) and (C3) it suffices to show that
$FDP_\infty(t_{\alpha,\infty}^\lambda)\leq \alpha$ if $\mu_0\leq 1$, with equality when $G(t)$ is a DU distribution with $\mu_0 =1$ and $FDP_{\infty}^\lambda(u)\geq \alpha$.  To show this, consider the following:
\begin{eqnarray*}
FDP_{\infty}(t_{\alpha,\infty}^\lambda) &=& a_0\mu_0\frac{t_{\alpha,\infty}^\lambda}{G(t_{\alpha,\infty}^\lambda)}\\
&\leq& a_0\frac{t_{\alpha,\infty}^\lambda}{G(t_{\alpha,\infty}^\lambda)}\\
&\leq& \frac{1-G(\lambda)}{1-\lambda}\frac{t_{\alpha,\infty}^\lambda}{G(t_{\alpha,\infty}^\lambda)} \\
&=& FDP^\lambda_{\infty}(t_{\alpha,\infty}^\lambda)\\
&\leq& \alpha.
\end{eqnarray*}
The first equality is due to the definition of $FDP_{\infty}(\cdot)$.  The first inequality is satisfied when $\mu_0\leq 1$ and is an equality when $\mu_0 = 1$.  As for the second inequality, note that $G(\lambda)\leq a_0\lambda + 1-a_0$ when $\mu_0\leq 1$ and $G(\lambda) = a_0\lambda + 1-a_0$ under a DU distribution with $\mu_0 = 1$.  Consequently
$$a_0 = \frac{1-[a_0\lambda + 1-a_0]}{1-\lambda} \leq \frac{1 - G(\lambda)}{1-\lambda}$$
when $\mu_0\leq 1$ and the inequality is an equality when $G$ is a DU distribution with $\mu_0 = 1$.  The last equality is satisfied by the definition of $FDP_{\infty}^\lambda(\cdot)$. The last inequality is satisfied by the definition of $t_{\alpha,\infty}^\lambda$ and is an equality when $G$ is a DU distribution with $\mu_0 = 1$ and $FDP_\infty(u)\geq \alpha$ because these conditions imply $FDP_\infty(u) = FDP_\infty^\lambda(u)\geq \alpha$.  That is, $FDP_\infty(u)$ is continuous and monotone and takes on value $\alpha$.  Hence, $FDP_{\infty}(t_{\alpha,\infty}^\lambda)\leq \alpha$ if $\mu_0\leq 1$ with equality if $G$ is a DU distribution with $\mu_0=1$ and $FDP_\infty(u)\geq \alpha$. $\|$\\

\noindent\textsc{Proof of Theorem 6.}
Under the conditions of the theorem
\begin{eqnarray*}
Cov(W_{m,M},\theta_{m,M}) &=& E[W_{m,M}|\theta_{m,M}=1]E[\theta_{m,M}] - E[W_{m,M}]E[\theta_{m,M}]\\
&=& E[\theta_{m,M}](E[W_{m,M}|\theta_{m,M}=1] - 1).
\end{eqnarray*}
Hence, $Cov(W_{m,M},\theta_{m,M})\geq 0$ implies $E[W_{m,M}|\theta_{m,M}=1]\geq 1$ and consequently $E[W_{m,M}|\theta_{m,M}=0]\leq 1$, with equality if $Cov(W_{m,M},\theta_{m,M}) = 0$.  Hence $E[\bar{W}_{0,M}|\bm{\theta}_M\neq \bm{1}_M] = \mu_0 \leq 1$ with equality if $Cov(W_{m,M},\theta_{m,M})=0$.  The result follows because $\bar{W}_{0,M}\rightarrow \mu_0$ almost surely. $\|$ \\

\noindent\textsc{Proof of Corollary 1.}
Observe that $u = 1$ and $\lambda<1$ is fixed.  Hence (A2) is satisfied and (A4) - (A6) are satisfied by the conditions of the Theorem.  Therefore Claim (C1) holds by Theorem 5.  Now, additionally note that $\mu_0 = 1$ if $\mat{w}_M = \bm{1}_M$ and that $FDP_{\infty}(1) = a_0 \geq \alpha$ under the conditions of the Theorem.  Thus Claims (C2) and (C3) hold by Theorem 5. $\|$ \\

Before proving Theorem 7 we provide Lemma S2.  It will be used to verify that optimal weights are weakly dependent so that decision functions satisfy the weak dependence structure defined in (A4) - (A5).  Below, denote $t_0(k) = E[t_m(k/p_m, \gamma_m)]$ and denote $G(t_0(k)) = E[\delta_m(t_m(k/p_m,\gamma_m))]$, where the expectations are taken over all random quantities, i.e. over $(Z_m,\theta_m,p_m,\gamma_m)$ for some fixed $k>0$.  Further, define
$$\widetilde{FDP}_\infty(t_0(k)) = \frac{1-G(t_0(k))}{1-t_0(k)}\frac{t_0(k)}{G(t_0(k))}$$
and $$k^* = \inf\{k:\widetilde{FDP}_\infty(t_0(k)) = \alpha\},$$
and denote $$\tilde{w}_{m,\infty}(k^*, p_m, \gamma_m) = U_m t_m(k^*/p_m,\gamma_m)/t_0(k^*).$$
\begin{lemma2}\label{as delta}
Suppose that $\Pr(p_m\leq 1-\alpha)$.  Under Model 1 and (A1), $k_M^* \rightarrow k^*$ almost surely and
$$\tilde{w}_{m,M}(k_M^*, \mat{p},\bm{\gamma}) \rightarrow \tilde{w}_{m,\infty}(k^*, p_m, \gamma_m)$$ almost surely.
\end{lemma2}
\begin{proof}
Note that $0<\alpha \leq 1 - p_{(M)}$ with probability 1 so that $k_M^*$ is well defined for $M = 1, 2, ...$ by Theorem 2.  Further, observe that $t_m(k^*/p_m,\gamma_m)$, $m = 1, 2, ...$ are i.i.d. continuous random variables taking values in $[0,1]$ under Model 1.  Hence, by the Strong Law of Large numbers $\bar{t}_M(k^*,\mat{p},\bm{\gamma}) \rightarrow t_0(k^*)$ almost surely.  Likewise, $\bar{G}_M(\bm{t}(k^*,\mat{p},\bm{\gamma})\rightarrow G(t_0(k^*))$ almost surely and by the Continuous Mapping Theorem $$\widetilde{FDP}_M(\mat{t}(k^*,\mat{p},\bm{\gamma}))\rightarrow   \widetilde{FDP}_\infty(t_0(k^*))$$ almost surely.  Because further $\widetilde{FDP}_M(\mat{t}(k,\mat{p},\bm{\gamma}))$ and $\widetilde{FDP}_\infty(t_0(k))$ are both continuous in $k$ by (A1), we have from the Continuous Mapping Theorem and the definitions of $k_M^*$ and $k^*$ that $k_M^*\rightarrow k^*$ almost surely.  Thus,
\begin{eqnarray*}
\tilde{\mat{w}}_{m,M}(k_M^*,\mat{p},\bm{\gamma}) &=& U_m\mat{w}_{m,M}(k_M^*,\mat{p},\bm{\gamma})\\
& = & U_m\frac{t_m(k_M^*/p_m,\gamma_m)}{\bar{t}_M(k_M^*,\mat{p},\bm{\gamma})} \\
&\rightarrow&U_m\frac{t_m(k^*/p_m,\gamma_m)}{t_0(k^*)} = \tilde{w}_{m,\infty}(k^*,p_m,\gamma_m)
\end{eqnarray*}
almost surely by the Continuous Mapping Theorem.

\end{proof}

\noindent\textsc{Proof of Theorem 7.}
First we verify (A2).  Observe $\lambda_M = \bar{t}_M(k_M^*,\mat{p},\bm{\gamma})\rightarrow t_0(k^*)$ almost surely by the Strong Law of Large Numbers and the Continuous Mapping Theorem, where recall $t_0(k^*) = E[t_m(k^*/p_m,\gamma_m)]$.  Thus, by the definition of $\tilde{w}_{m,M}$
$$\lim_{M\rightarrow\infty}\tilde{w}_{m,M} = \lim_{M\rightarrow\infty}\frac{U_m t_{m,M}(k_M^*/p_m,\gamma_m)}{\bar{t}_M(k_M^*,\mat{p},\bm{\gamma})}\leq \frac{1}{t_0(k^*)}$$ almost surely, where the last inequality is due to the Continuous Mapping Theorem, Lemma (S2) and because $U_mt_m(k_M^*,\mat{p},\bm{\gamma})\leq 1$ almost surely by construction.  That is, (A2) is satisfied with $\lambda = u = 1/t_0(k^*)$.

Before verifying (A4) - (A6) we introduce some notation. Denote
$$G^{k^*}(t) = E[\delta_m(t\tilde{w}_{m,\infty}(k^*,p_m,\gamma_m))]$$
where the expectation is taken over all random quantities, i.e. taken over $(Z_m, \theta_m,p_m,\gamma_m,U_m)$.  Further we sometimes suppress $\mat{p}$ and $\bm{\gamma}$ and write $\tilde{w}_{m,\infty}(k^*) = \tilde{w}_{m,\infty}(k^*,p_m,\gamma_m)$, $\tilde{w}_{m,M}(k^*) = \tilde{w}_{m,M}(k^*,\mat{p},\bm{\gamma})$ and $\tilde{\mat{w}}_M(k^*) = [\tilde{w}_{m,M}(k^*), m\in\mathcal{M}]$.  Further, denote $\tilde{\mat{w}}_\infty(k^*) = [\tilde{w}_{m,\infty}(k^*), m\in\mathcal{M}]$.

Now consider (A4). Observe that $\delta_m(t\tilde{w}_{m,\infty}(k^*)), m = 1, 2, ...$ are i.i.d. Bernoulli random variables with success probability $G^{k^*}(t)$ under Model 1 so that
$$\frac{R(t\tilde{\mat{w}}_{\infty}(k^*))}{M} = \frac{\sum_{m\in\mathcal{M}}\delta_m(t\tilde{w}_{m,\infty}(k^*))}{M}\rightarrow G^{k^*}(t)$$ almost surely by the Strong Law of Large Numbers.  Further, by the NS assumptions, Lemma S2, and because $G^{k^*}(t)$ is continuous, we have that for any $\epsilon>0$ there exists an $\epsilon'>0$ such that
\begin{eqnarray*}
\frac{R(t\tilde{\bm{w}}_M(k^*_M))}{M} &=& \frac{\sum_{m\in\mathcal{M}}\delta_m(t\tilde{w}_{m,M}(k_M^*))}{M}<\frac{\sum_{m\in\mathcal{M}}\delta_m(t[\tilde{w}_{m,\infty}(k^*)+\epsilon'])}{M} \\ &<& G^{k^*}(t+t\epsilon')<G^{k^*}(t) + \epsilon
\end{eqnarray*}
for all large enough $M$.  Similar arguments give
$$
\frac{R(t\tilde{\bm{w}}_M(k^*_M))}{M}> G^{k^*}(t) - \epsilon
$$
for all large enough $M$.  Hence, $R(t\tilde{\mat{w}}_{M}(k^*))/M \rightarrow G^{k^*}(t)$ almost surely.  Then because $k_M^*\rightarrow k^*$ almost surely by Lemma S2, $R(t\tilde{\mat{w}}_{M}(k_M^*))/M \rightarrow G^{k^*}(t)$ almost surely by the Continuous Mapping Theorem.

As for (A5), recall the NS conditions give $E[\delta_m(t_m)|\theta_m = 0] = t_m$.  Hence, taking the expectation over all random quantities, we have by the law of iterated expectation
$$E[(1-\theta_m)\delta_m(t \tilde{w}_{m,\infty}(k^*,p_m,\gamma_m))] = a_0\mu_0t,$$
where $a_0 = E[1-\theta_m]$ and $\mu_0 = E[\tilde{w}_{m,\infty}(k^*, p_m,\gamma_m)|\theta_m = 0]$.  Then, arguments akin to those in the proof of (A4) give
$$\frac{V(t\tilde{\mat{w}}_{M}(k_M^*))}{M} = \frac{M_0}{M}\frac{V(t\tilde{\mat{w}}_{M}(k_M^*))}{M_0} \rightarrow a_0 \mu_0 t$$
almost surely.

For (A6), first observe that $G^{k^*}(t) = a_0\mu_0 t + (1-a_0) G_1(t)$ for $t\leq u$,
where $$G_1(t) = E\left[\pi_{{\gamma}_m}(t\tilde{w}_{m,\infty}(k^*,p_m,\gamma_m))\right]$$ and the expectation is taken over all random quantities.  Clearly $t\mapsto G_1(t)$ is concave and twice differentiable because $t\mapsto \pi_{\gamma_m}(t)$ is twice differentiable almost surely by (A1).
To see that $t/G(t)\rightarrow 0$ as $t\downarrow 0$ note that $G_1'(t)\rightarrow \infty$ as $t\downarrow 0$ because $\pi_{\gamma_m}'(t)\rightarrow \infty$ as $t\downarrow 0$ almost surely by (A1).  Hence,
$$\frac{t}{G^{k^*}(t)} = \frac{t}{a_0 \mu_0 t + (1-a_0)G_1(t)}\rightarrow 0$$
as $t\downarrow0$ by an application of H{\^o}ptial's rule.  Clearly, $\lim_{t\uparrow u}t/G^{k^*}(t) \rightarrow u/G^{k^*}(u)$ because $G^{k^*}(t)$ is continuous. To see that $u/G^{k^*}(u) \leq 1$ we establish the following:
\begin{eqnarray*}
G^{k^*}(u)&=& E[\delta_m(u\tilde{w}_{m,\infty}(k^*))]\\
&=& a_0E[\delta_m(u\tilde{w}_{m,\infty}(k^*))|\theta_m=0] \\
&& + (1-a_0)E[\delta_m(u\tilde{w}_{m,\infty}(k^*))|\theta_m=1] \\
&=& a_0E[u\tilde{w}_{m,\infty}(k^*)] + (1-a_0) E[\pi_{\gamma_m}(u\tilde{w}_{m,\infty}(k^*))] \\
&\geq& a_0E[u\tilde{w}_{m,\infty}(k^*)] + (1-a_0)E[u\tilde{w}_{m,\infty}(k^*)] \\
&=& E[u\tilde{w}_{m,\infty}(k^*)] \\
&=&uE[\tilde{w}_{m,\infty}(k^*)] = u.
\end{eqnarray*}
The first equality is by the definition of $G^{k^*}(u)$ while the second equality is due to the law of iterated expectation. The third is a consequence of the definition of $\pi_{\gamma_m}(t)$ and the NS assumptions.  The inequality is satisfied because $\pi_{\gamma_m}(t)\geq t$ almost surely for every $t\in[0,1]$ under (A1). The forth equality is obvious.  As for the fifth, recall $E[U_m|p_m,\gamma_m] = 1$, $\tilde{w}_{m,\infty}(k^*) = U_m w_{m,\infty}(k^*)$ and that $E[w_{m,\infty}(k^*)] = 1$.  Hence, by the law of iterated expectation $E[\tilde{w}_{m,\infty}(k^*)] = E[w_{m,\infty}(k^*)] = 1.$

To verify that $\mu_0\leq 1$ we make use of Theorem 6 and write $W_{m} = w_{m,M}(k_M^*,\mat{p},\bm{\gamma})$ and $\tilde{W}_m = U_m W_m $ for short.  First let us focus on $Cov(W_m,\theta_m)$.   From the law of iterated expectation,
\begin{equation}\label{cov}
Cov(W_{m}, \theta_m) = E[Cov(W_{m},\theta_m|p_m)] + Cov(E[W_{m}|p_m],p_m).
\end{equation}
Observe that
\begin{eqnarray*}
Cov(W_{m},\theta_m|p_m) & = & E[W_{m}\theta_m|p_m] - E[W_{m}|p_m]E[\theta_m|p_m]\\
 &=& p_mE[W_{m}|p_m] - p_mE[W_{m}|p_m] = 0
\end{eqnarray*}
which implies that the first expectation in (\ref{cov}) is 0. To compute the second expectation, first observe $\pi_{\gamma_m}'(t_m)$ is continuous and strictly decreasing and hence the solution to $\pi_{\gamma_m}'(t_m) = a$, denoted $t_m(a,\gamma_m)$, is continuous and strictly decreasing in $a$ almost surely by (A1).  Hence $t_m(k_M^*/p_m, \gamma_m)$ is strictly increasing and continuous in $p_m$ almost surely.  Thus,
$$E[W_{m}|\mat{p},\bm{\gamma}] = E\left[M\frac{t_m(k_M^*/p_m,\gamma_m)}{t_m(k_M^*/p_m,\gamma_m) + \sum_{j\neq m}t_j(k_M^*/p_j,\gamma_j)}\bigg|\mat{p},\bm{\gamma}\right]$$ is also increasing in $p_m$ almost surely because the function $x/(x + a)$ for $a$ any positive constant is increasing in $x$ for $x>0$.  This implies $E[W_{m}|p_m]$ is also increasing in $p_m$ almost surley which implies $Cov(E[W_{m}|p_m],p_m)\geq 0$. As for $\tilde{W}_{m} = U_mW_{m}$,
$$Cov(\tilde{W}_m,\theta_m) = E[Cov(U_mW_m,\theta_m|W_m)] + Cov(E[U_mW_m|W_m], E[\theta_m|W_m])$$
by the law of iterated expectation.  But $$E[Cov(U_mW_m,\theta_m|W_m)] = E[W_mCov(U_m,\theta_m|W_m)] = 0$$ because $Cov(U_m,\theta_m|W_m)$ is 0 by construction.  Additionally, $$Cov(E[U_mW_m|W_m], E[\theta_m|W_m]) =Cov(W_m,E[\theta_m|W_m])\geq 0$$ because $Cov(W_m,\theta_m)\geq 0$.  Hence, $Cov(\tilde{W}_m,\theta_m)\geq 0$ and thus, by Theorem 6, $\mu_0\leq 1$.  $\|$ \\

\noindent\textsc{Proof of Theorem 8.}
First recall from the proof of Theorem 7 (where here we take $U_m = 1$ almost surely for every $m$) that $\lambda_M = \bar{t}_M(k_M^*)\rightarrow t_0(k^*)$ Hence,  we have
$$FDP_{\infty}^{\lambda}(t) = \frac{1 - G^{k^*}(t_0(k^*))}{1 - t_0(k^*)}\frac{t}{G^{k^*}(t)}.$$
Further observe that because $t/G^{k^*}(t)$ is strictly increasing by (A6), then $t_0(k^*) = t_{\alpha,\infty}^\lambda$ by the definition of $t_{\alpha,\infty}^\lambda$.  Hence $\bar{t}_M(k_M^*)\rightarrow t_0(k^*) = t_{\alpha,\infty}^\lambda$ almost surely. $\|$\\


\noindent\textsc{Proof of Corollary 2.}
First observe that $Cov(w_{m,M},\theta_{m,M}) = 0$ and hence $\mu_0 = 1$ by Theorem 6.  It therefore suffices to show that (A4) - (A6) are satisfied. But $\delta_m(tw_{m,M})$, $m = 1, 2, ...$ are i.i.d. Bernoulli random variables under Model 1 and the conditions of the Theorem.  Hence,
$R(t\mat{w}_M)/M \rightarrow G(t)$ for $G(t) = E[\delta_m(t\mat{w}_{m,M})]$ almost surely by the Strong Law of Large Numbers and (A4) is satisfied.  Likewise $(1-\theta_{m,M})\delta_m(tw_{m,M})$, $m = 1, 2, ...$ are i.i.d. random variable with mean $a_0 t$ under the NS assumptions and the conditions of the Theorem.  Hence,
$$\frac{V(t\mat{w}_M)}{M} = \frac{1}{M}\sum_{m\in\mathcal{M}}(1-\theta_{m,M})\delta_m(tw_{m,M})\rightarrow a_0 t$$ almost surely by the Strong Law of Large Numbers and (A5) is satisfied.  Condition (A6) is verified using arguments identical to those used in verifying (A6) in the proof of Theorem 7 with $G^{k^*}(t) = G(t)$ and $w_{m,M} = \tilde{w}_{m,\infty}(k^*)$. $\|$ \\

\noindent\textsc{Proof of Corollary 3.}
Observe that $\pi(\mat{t},\mat{p},\bm{\gamma})$ is proportional to  $\pi(\mat{t},\mat{1},\bm{\gamma})$ and hence the maximization of $\pi(\mat{t},\mat{p},\bm{\gamma})$ with respect to $\mat{t}$ is free of $p_m$.  Thus $\tilde{w}_{m,M}(k,\mat{p},\bm{\gamma})$ is independent of $p_m$ and hence independent of $\theta_m$.  The result then follows from Theorems 6 and 7. $\|$

\section{Title of section 4}
\setcounter{equation}{0}

\section{Simulation details}
\subsection{Simulations 1 - 4}

\begin{table}[h!]\center
\caption{\label{simul table} The average CDP (FDP) for the UU, UA,  WU, and WA procedures in Simulations  1 - 4.}
\begin{tabular}{cccccccc}
\multicolumn{4}{c} {Simulation 1}  \\ \cline{1-4}

& a=1 & a=3 & a=5 \\ \cline{2-4}
UU &0.007(0.021) &0.390(0.025) & 0.709(0.025)  \\
WU &0.007(0.021) &0.397(0.025) & 0.731(0.025)  \\
UA &0.007(0.021) &0.437(0.031) & 0.761(0.039)  \\
WA &0.007(0.021) &0.442(0.032) & 0.793(0.039)  \\ \\

 \multicolumn{4}{c} {Simulation 2} \\ \cline{1-4}

& a=1 & a=3 & a=5 \\ \cline{2-4}

UU & 0.007(0.024) & 0.390(0.025) & 0.709(0.025)\\
WU &  0.011(0.008) & 0.434(0.014) & 0.756(0.016)\\
UA & 0.007(0.024) & 0.430(0.031) & 0.761(0.041) \\
WA & 0.011(0.008) & 0.473(0.018) & 0.814(0.026) \\ \\

\multicolumn{4}{c} {Simulation 3} \\ \cline{1-4}

&a=1 & a=3 & a=5\\ \cline{2-4}

UU & 0.007(0.023)& 0.391(0.025) & 0.709(0.025)\\
WU & 0.013(0.007) & 0.404(0.015) & 0.719(0.016) \\
UA & 0.007(0.023)& 0.430(0.031) & 0.757(0.039)  \\
WA & 0.013(0.007)& 0.439(0.019) & 0.774(0.027) \\ \\

\multicolumn{4}{c} {Simulation 4} \\ \cline{1-4}
&a=1 & a=3 & a=5 \\ \cline{2-4}
UU & 0.007(0.025)& 0.391(0.025) & 0.710(0.025)\\
WU & 0.006(0.023) & 0.354(0.025) & 0.682(0.025) \\
UA & 0.007(0.025)& 0.425(0.030) & 0.756(0.039) \\
WA & 0.006(0.023)& 0.387(0.030) & 0.727(0.039) \\ \hline \\
\end{tabular}
\end{table}

In Simulation 1, observe that the FDP is increasing in $a$ for both adaptive procedures.   For example the FDP of each procedure is $0.021$ when $a=1$ but is $0.039$ when $a = 5$.  This is to be expected as both adaptive procedures are $\alpha$-exhaustive (see Corollaries 1 and 3) and hence we should expect the FDP to be near 0.05 in high power settings, i.e. for large $a$. Additionally, the largest gain in power (in terms of the average CDP) of the weighted adaptive procedure over the unweighted adaptive procedure occurs when effect sizes are most heterogeneous.  When $a = 5$ the average CDP of the WA procedure is $0.793$ while the average CDP of the UA procedure is $0.761$.  When data are homogeneous (a = 1), the CDPs of the procedures are identical.

In Simulation 2, data generating mechanisms are even more heterogeneous as now the $p_m$s also vary. General conclusions regarding the CDP are the same, with the advantages  of the weighted procedures over their unweighted counterparts being more pronounced. For example, the average CDP of the WAMDF for $\gamma_m\stackrel{i.i.d.}\sim Un(1,5)$ increased from 0.793 to 0.814 when allowing $p_m$s to vary, while for the UA procedure the CDP is still 0.761. We also observe that for $a = 5$ the average FDP of the WA procedure is only $0.026$ while the average FDP of the UA procedure is closer to 0.05; it is 0.039.  This is to be expected because, even though the WAMDF will dominate the UA procedure in terms of the average CDP, the UA procedure is $\alpha$-exhaustive while the WAMDF need not be in this setting.

Now consider non-optimal weights in Simulations 3 and 4.  \cite{RoeWas09} concluded that, in the unadaptive setting (the UU and WU procedures), weighted MDFs are robust with respect to weight misspecification in that they generally yield about as many or more rejected null hypotheses as unweighted procedures as long as weights are ``reasonably guessed'' and yield slightly less rejected null hypotheses when weights are poorly guessed. Simulations 3 and 4 confirm their results and further illustrates that the robustness property applies to adaptive procedures.  For example, comparing the unadaptive procedures in Simulation 3, we see that the average CDP of the WU(UU) procedures are 0.013(0.007), 0.404(0.391), and 0.719(0.709) for $a = 1, 3, 5$, respectively.  The average CDP of the WA(UA) procedure is 0.013(0.007), 0.439(0.430), and 0.774(0.757), for $a = 1, 3, 5$, respectively.  That is, when weights are positively correlated with optimal weights, weighted procedures still perform slightly better than their unweighted counterparts.  In the worst case scenario setting in Simulation 4, where weights are independently generated, the FDP is still controlled by the WA procedure, but some loss in power over its unweighted counterpart is observed.  For example, the CDP of the WA(UA) procedure is 0.006(0.007), 0.386(0.425), and 0.727(0.756) when $a = 1, 3,$ and  5, respectively, while the average FDP of the WA(UA) procedure is 0.025(0.023), 0.030(0.030), and 0.039(0.039) when $a = 1, 3, $ and $5$, respectively.  

\subsection{Simulation 5}
The average CDP ratio (weighted/unweighted) vs. the average CDP of the weighted procedure is depicted in Figure \ref{simulationFig} for all settings. Observe that the CDP ratio is greater than or equal to 1 for each value of $p$ and $\alpha$ as long as the CDP is at least 0.2.  \\

\begin{figure}[h!]\centering \vspace{-.5in}
\epsfig{file = simulationplot, width = 3.5in, height = 2.2in}
\caption{\label{simulationFig} The ratio of the average CDP (weighted/unweighted) vs. the average CDP of the weighted procedure for $p = 0.2 (o),
p= 0.5 (\triangle),$ and $p = 0.8 (+)$ for $\bar{\gamma} = 1.75, 2, 2.25$.}
\end{figure}